\newcount\Comments  
\documentclass[11pt]{article}
\usepackage{booktabs} 
\usepackage[ruled]{algorithm2e} 

\usepackage[margin=1in]{geometry}
\usepackage[utf8]{inputenc} 
\usepackage[T1]{fontenc}
\usepackage{amsmath}
\usepackage{amssymb}
\usepackage{amsfonts}
\usepackage{amsthm}
\usepackage{natbib}
\usepackage{bm}
\usepackage{comment} 
\usepackage{subcaption}
\usepackage{pgf,pgfplots}
\pgfplotsset{compat=1.14}
\usepackage{mathrsfs}
\usetikzlibrary{arrows}

\newcommand{\E}{\mathbb{E}}

\newtheorem{definition}{Definition}[section]
\newtheorem{lemma}{Lemma}[section]
\newtheorem{theorem}{Theorem}[section]
\newtheorem{claim}{Claim}[section]
\newtheorem{corollary}{Corollary}[section]
\newtheorem{fact}{Fact}[section]
\newtheorem{example}{Example}[section]
\newtheorem{proposition}{Proposition}[section]

\newcommand{\ignore}[1]{}
\newcommand{\kibitz}[2]{\ifnum\Comments=1\textcolor{#1}{#2}\fi}
\newcommand{\sherry}[1]{\kibitz{red}{\noindent[Sherry: #1]}}

\title{Optimal Advertising for Information Products\thanks{This work is supported by the National Science Foundation under Grant No. CCF-1718549. }}

\author{
Shuran Zheng \\
Harvard University\\
{\normalsize shuran\_zheng@seas.harvard.edu}\\
\and
Yiling Chen \\
Harvard University\\
{\normalsize yiling@seas.harvard.edu}\\
}
\date{September 2021\\Revised after EC'21\footnote{We would like to thank the EC'21 participants who came to our poster session for the helpful discussion. We are extremely grateful to Kevin He for informing us of the important missing references in Bayesian Persuasion. }}

\begin{document}

\begin{titlepage}

\maketitle

\begin{abstract}
When selling information products, sometimes the seller can provide some free partial information to change people's valuations so that the overall revenue can possibly be increased. 
In this work, we study the general problem of advertising information products by revealing  partial information. We consider buyers who are decision-makers. The outcomes of the decision problems depend on the state of the world that is unknown to the buyers.
The buyers can make their own observations and thus can hold different personal beliefs about the state of the world. There is an information seller who has access to the state of the world. The seller can promote the information by revealing some partial information. We assume that the seller chooses a long-term advertising strategy and then commits to it. The buyers decide whether to purchase the full information product after seeing the partial information. The seller's goal is to maximize the expected revenue. 
 We study the problem in two settings. 
 \begin{enumerate} 
 \item	The seller targets buyers of a certain type. In this case, finding the optimal advertising strategy is equivalent to finding the concave closure of a simple function. The function is a product of two quantities. The first one is \emph{the likelihood ratio}~\citep{ALONSO2016672}, which captures how much the buyer's personal belief deviates from the prior. The second one is \emph{the cost of uncertainty}~\citep{frankel2019quantifying}, which represents the value of the information to the buyer. 
 Based on this observation, we prove some properties of the optimal mechanism, which allow us to solve for the optimal mechanism by a finite-size convex program. The convex program will have a polynomial size if the state of the world has a constant number of possible realizations or the buyers face a decision problem with a constant number of options. For the general problem, we prove that it is NP-hard to find the optimal mechanism. 
 \item For the general problem when the seller faces buyers of different types and only knows the distribution of their types, we provide an approximation algorithm that finds an $\varepsilon$-suboptimal mechanism when it is not too hard to predict the possible type of buyers who will make the purchase. For the general problem, we prove that it is NP-hard to find a constant-factor approximation.
 \end{enumerate}
\end{abstract}

\end{titlepage}


\section{Introduction}

The trading of information constitutes an increasingly important business in modern economies. The rapid spread of the internet in the past few decades has provided easy access to a large volume of online data, which stimulated the dynamically growing markets for information. Information is being sold in a large variety of forms: online newspapers and magazines, consulting services, database access, industry reports and credit reports, etc.

The nature of information products varies greatly from traditional commodity products. A lot of effort has been made to understand the optimal strategy of selling information. It has been long observed that revealing partial information about the information products may greatly increase the subsequent likelihood of purchase: movies have trailers, online newspapers and magazines provide free beginning paragraphs or pages, dataset platforms allow the potential buyers to browse the datasets and give free random samples of the data, etc. Revealing partial information decreases the amount of information that is later being sold, but in return, it may change people's opinion about the product and increase some buyers' interest in the information product.


In this work, we study the problem of promoting information products by revealing some partial information. 
For example, consider a flight tracker which wants to sell the information of flight delays to travelers. In this problem, the information being sold (or the state of the world) is the flight delay. To promote the information of flight delay, the flight tracker can send out some partial information, e.g. whether the delay is longer than six hours, which is a signal correlated with the state of the world. If a traveler purchases the information, he may use it to update his travel plan. The information's value for a traveler is determined by his expected gain in replanning.


We consider an information seller who has access to the state of the world which is valuable to some decision-makers. 
We consider a long-term seller who needs to decide an advertising strategy (e.g. tell the travelers whether the delay is longer than six hours, provide a free random sample of the datasets) and follow that strategy thereafter. We assume that the seller's advertising strategy is publicly known.

Each day, some buyers will arrive. Each of the buyers faces a decision problem, the outcome of which depends on the state of the world and the action he takes.
The buyers do not have direct access to the state of the world, but they may have partial observations (e.g. the weather) and thus hold personal beliefs about the state of the world when they arrive.
The buyers decide whether to purchase the full information after seeing the partial information provided by the seller. For example, a buyer decides whether to pay $5$ dollars for the exact delay after knowing that the delay is longer than six hours. We assume that the buyers are rational and they perform Bayesian updating on their beliefs after seeing the partial information revealed by the seller.

It is worth noting that in this work we assume the buyers can have personal beliefs about the state of the world. This deviates from one of the ubiquitous but controversial assumptions in economic theory: the common prior assumption. In models of asymmetric information, the common prior assumption is that there is an ex-ante stage at which the individuals have identical information and subsequently update their beliefs in response to private signals. The plausibility of assuming common priors has been questioned and discussed (see~\citep{morris1995common}). In this work, we do not assume common priors always exist but consider the individuals' beliefs about the external world to be the primitives of the model. Nevertheless, the case when a common prior exists is just a special case of our model. 

In addition, our work focuses on the interaction between the seller and the buyer through the advertising strategy but omit the consideration of advertising cost. We assume that the cost difference between different advertising strategies is relatively small compared to the seller's revenue.

We will consider two settings: (1) The seller targets a group of buyers of a certain type. This can be applied to the case when the majority of the buyers share a common belief and have the same goal. We study this simplified situation to understand the hardness of the problem and gain some insight into the optimal advertising strategy. (2) The seller faces buyers of different types and only knows the distribution of their types.

\subsection{Our Results}
Our first important observation is that partial information disclosure can be beneficial for the seller when the buyers have personal beliefs. 
Based on this observation, we study the optimal advertising strategy for an information seller.

 We first consider a seller who targets a group of  buyers of a specific type.
In this case, we apply the results from~\citep{ALONSO2016672} to show that finding the optimal advertising strategy is equivalent to  finding the concave closure of a simple function. We present a finite-size convex program that solves the optimal mechanism. Furthermore, when the state of the world has only a few possible realizations (e.g. the weather is going to be sunny/cloudy/rainy...), or the buyers face a decision problem with only a few options,  the convex program will have a polynomial size and thus can be efficiently computed.
\begin{theorem}[Informal]
When the seller targets a group of buyers of a specific type, solving the optimal advertising mechanism is equivalent to finding the concave closure of a simple function. The optimal mechanism can be solved by a finite-size convex program.
When the state of the world has a constant number of possible realizations, or the buyers face a decision problem with a constant number of options, the convex program will have a polynomial size. 
\end{theorem}
The function, whose concave closure indicates the optimal advertising mechanism, is the product of two quantities. The first one is what we call \emph{the likelihood ratio function}, which depends on the buyer's personal belief about the state of the world. The second component is \emph{the cost of uncertainty function}, which represents the value of the information to the buyer. 

In addition, our convex program shows that there exists an optimal mechanism that reveals partial information in a way that the buyer will be able to reduce the range of the state to a set of size $\le 2|A|$, where $A$ is the buyer's action set.
\begin{theorem}[Informal] \label{thm:inf_insight}
There exists an optimal mechanism that guarantees that the buyer's posterior about the state of the world has no more than $2|A|$ non-zero entries after seeing the partial information.
\end{theorem} 
For the general problem, our convex program can possibly be exponentially large. We show this is not surprising because the problem is NP-hard. 
 \begin{theorem}[Informal]
 	When the seller targets a group of buyers of a specific type, solving the optimal advertising mechanism is NP-hard.
 \end{theorem}
 
When the seller faces buyers of different types and only knows the distribution of their types, the problem becomes more challenging because it is more difficult for the seller to choose the best price (after the advertising). 
Nevertheless, we show that in some special cases when it is not too hard to predict the possible types of buyers who will make the purchase, it is possible to find an $\varepsilon$-suboptimal mechanism  by a linear program.
\begin{theorem}[Informal]
	When the buyers' types are drawn from a known distribution, and the set of buyer types that will finally purchase the information has polynomially many possibilities, we can find an $\varepsilon$-suboptimal mechanism  within running time polynomial in $1/\varepsilon$ and the input size.
\end{theorem}
Furthermore, this $\varepsilon$-suboptimal mechanism can be solved by an LP.
For the general problem, the optimal mechanism is not only hard to solve, but also hard to approximate. 
\begin{theorem}[Informal]
When the buyers' types are drawn from a known distribution, it is NP-hard to find a constant-factor approximation  for our optimal information advertising problem.
\end{theorem}

 \subsection{Related Work}
Markets for information and data have attracted an increasing amount of attention recently. 
 We refer the readers to~\citep{bergemann2019markets} for an overview of the vast literature. In parallel with the analysis of competitive markets of information (see~\citep{sarvary2011gurus} for an overview) and the study of data intermediaries~\citep{bergemann2019markets,bergemann2019economics}, our work falls into the category of a monopoly information holder directly selling information to the buyers.  In contrast to some works that focus on specific information product, e.g. selling cookies \citep{bergemann2015selling} and selling datasets \citep{mehta2019sell}, we consider  selling information in a general framework, which makes our work most relevant to \citep{esHo2007price,bergemann2018design,babaioff2012optimal,chen2020selling,cai2020sell}. What makes our work different from the previous ones is that we consider a seller who can only use posted price mechanism with a single price. The previous works  \citep{esHo2007price,bergemann2018design,babaioff2012optimal,chen2020selling,cai2020sell} all consider designing a menu of different information with different prices. The size of the optimal menu is as large as the type space (due to the use of the revelation principle). \citet{babaioff2012optimal} and \citet{chen2020selling} actually consider a seller that can interact with the buyer in multiple rounds. Although larger mechanisms give the seller more power to extract revenue, they are also more difficult to implement and participate in. Therefore in this work, we consider the design of simple mechanisms for selling information, in which the seller just posts a price for the full revelation of information (menu size equal to one), but can partially reveal some relevant information before the sale to promote the information product.
 
 A particularly relevant topic is Bayesian persuasion~\citep{kamenica2011bayesian, kamenica2018bayesian,Dughmi2017}, especially the public persuasion problem~\citep{dughmi2019hardness,xu2020tractability} and Bayesian persuasion with heterogeneous priors~\citep{alonso2016persuading, ALONSO2016672}. In Bayesian persuasion, there is a sender and a receiver. The sender wants to persuade the receiver to take some actions by choosing a signal (or in our words, choosing some partial information) to reveal to the receiver. Our problem is very close to the public persuasion problem in the sense that the seller sends a public signal to persuade the buyers. The key difference between our problem and the public persuasion problem is that in Bayesian persuasion, the sender only decides the signaling scheme that is used to persuade the receivers; but in our problem, the seller also needs to choose a price menu. 
 But in this work, we mainly focus on the case when the seller targets buyers with a specific belief. In this case, the optimal price menu can be immediately decided for a chosen signaling scheme. As a result, our problem becomes a Bayesian persuasion problem with heterogeneous priors~\citep{ALONSO2016672}. But computing the optimal mechanism for this problem is still not trivial. 
 We also want to point out a work~\citep{rayo2010optimal} that studies a quite different information disclosure problem but has a very close underlying mathematical model. Actually, their problem can be seen as a special case of ours. We discuss this in Appendix~\ref{app:optimal_inf_dis}.
 
 There is also a vast recent literature on information design that studies how different information disclosure rules influence the outcomes of games in different settings (see~\citep{bergemann2019information}),  including the information disclosure in pricing~\citep{rayo2010optimal,smolin2019disclosure, ali2020voluntary}, in auctions~\citep{Miltersen12,Badanidiyuru2018,Daskalakis2016does,emek2014signaling,Peter2007optimal}, in two-sided markets~\citep{romanyuk2019cream,johari2019quality,bimpikis2020information}, in normal-form games~\citep{bhaskar2016hardness,cheng2015mixture,dughmi2014hardness}, etc.

 
It is worth noting that there is a fundamental difference between advertising regular goods (see~\citep{bagwell2007economic} for an overview) and advertising information products. 
 Providing additional information about regular goods will not make any change to the goods themselves. But advertising information products may change the information product itself as revealing relevant information may decrease the amount of information that is finally being sold.
\section{Model}
We consider the setting with a monopolist information seller and information buyers who need to make a decision based on the information held by the seller. The information being sold is the state of the world  $\omega \in \Omega = \{1, \dots, n\}$, which is drawn from a \emph{commonly known} distribution $\mu(\omega)$.

Each day, a new state of the world $\omega$ will be realized and some information buyers will come. The information buyers need to choose an action $a\in A$. A buyer's utility $u(\omega, a)$ depends on his action $a$ and the state of the world $\omega$ on that day.  The buyers cannot directly observe $\omega$ of that day, but they may have their own partial observations  and thus may hold personal beliefs about the state of the world (which can be different from $\mu(\omega)$). We denoted by $\theta \in \Theta \subseteq \Delta \Omega$ the buyer's personal belief, which is a distribution over $\Omega$ with full support. We assume that the set of possible personal beliefs $\Theta \subseteq \Delta \Omega$ is a finite set. In the work, we also call $\theta$ the \emph{type} of the buyer. Without loss of generality, we assume that the utility function is normalized so that $u(\omega, a) \in [0,1]$.\footnote{The case that different types of buyers have different action sets and different utility functions can be converted into a single action set and a common utility function by merging each buyer's action sets and the associated utility functions. So without loss of generality we assume there is a single action set and a common utility function. }


The information seller has access to the realized state of the world $\omega$ every day. The seller needs to decide a long-term strategy to sell the information of $\omega$.  
 We assume the seller can only sell the information by a posted price mechanism with a single price, that is, set a price for telling the buyers the value of $\omega$. But before selling the information, the seller can advertise the information of $\omega$ by sending out some partial information, or more formally, the seller can send a \emph{signal} that is correlated with the state of the world. The buyers will update his belief about the state of the world $\omega$ after seeing the signal. The seller then post a price for the full revelation of $\omega$. The price can be different when the buyers see different signal realizations.
Formally, the seller can use an advertising rule defined as follows.
\begin{definition}
An advertising rule $\langle S, \pi, \{p_s : s \in S\} \rangle$ consists of 
\begin{itemize}
	\item a finite set of signals $S$,
	\item  a signaling scheme $\pi$, which is a random mapping from the support of the state of the world $\Omega$ to the signals $S$, i.e., $\pi: \Omega \to \Delta S$,
	\item and a price menu $\{p_s : s \in S\}$.
\end{itemize}	 
When using advertising rule $\langle S, \pi, \{p_s : s \in S\} \rangle$, the seller will first send a signal $s\in S$ by the signaling scheme $\pi$, that is, when the state of the world is $\omega$ the seller will send signal $s\in S$ with probability $\pi(\omega, s)$. Then if the signal that has been sent is $s$, the seller will charge price $p_s$ for the full revelation of $\omega$. 
\end{definition}

\begin{example}
	In the example of selling flight delay, the seller sends two possible signals $$S = \{\text{below $6$ hours}, \text{above $6$ hours}\}$$ with signaling scheme
	\begin{align*}
			\pi(\omega, \text{below $6$ hours}) = \bm{1}(\omega \le 6), \\
			\pi(\omega, \text{above $6$ hours}) = \bm{1}(\omega > 6).
	\end{align*}	
\end{example}

\paragraph{Buyer strategy.} Consider a buyer with personal belief $\theta = (\theta_1, \dots, \theta_n)$ before seeing the signal $s$, then when the signal is realized to $s$, the posterior belief of the buyer will be
\begin{align} \label{eqn:posterior_eta}
\eta^s (\theta) = \frac{\big( \theta_1 \pi(1, s), \dots, \theta_n \pi(n, s)\big)}{\sum_{\omega=1}^n \theta_\omega \pi(\omega,s)}.
\end{align}
Then the highest price the buyer is willing to pay for the full revelation of $\omega$ will be his expected loss of not knowing $\omega$ based on his posterior belief $\eta^s (\theta)$. This expected loss (as a function of $\eta$) is defined as the \emph{cost of uncertainty}  by~\cite{frankel2019quantifying}. Here we use the same term. 
\begin{definition}[Cost of uncertainty~\citep{frankel2019quantifying}] \label{def:cost_of_unc}
For a decision maker with utility function $u(\omega, a)$ and a belief $\eta = (\eta_1, \dots, \eta_n) \in \Delta \Omega$, the cost of uncertainty is equal to the expected loss of not knowing $\omega$,
\begin{align*}
C(\eta) & =   \E_{\omega \sim \eta} \big[\max_{a\in A} u(\omega, a) \big] - \max_{a\in A} \E_{\omega \sim \eta} [ u(\omega, a) ] \\
& =   \sum_{\omega=1}^n \eta_\omega \max_{a\in A} u(\omega, a) - \max_{a\in A} \sum_{\omega = 1}^n \eta_\omega u(\omega,a)\\
& = \min_a\, C_a(\eta),  
\end{align*}
where $C_a(\eta) =  \sum_{\omega=1}^n \eta_\omega \left(\max_{a'\in A} u(\omega, a') -  u(\omega,a)\right)$ is a linear function of $\eta$ that represents the expected regret of taking action $a$. Since $C(\eta)$ is the minimum of $|A|$ linear functions, $C(\eta)$ is a concave function.
\end{definition}
So when the signal is realized to $s$, the buyer will purchase the full revelation of $\omega$ if and only if his expected gain of knowing $\omega$ based on his posterior belief is higher than the price, $C(\eta^s(\theta)) \ge p_s$.
 
\paragraph{Mechanism design problem.} We assume that the seller knows the utility function $u(\omega, a)$, $\mu(\omega)$, and the conditional distribution of buyers' personal beliefs $\mu(\theta|\omega)$.\footnote{In this work, we do not assume common priors always exist, but consider the individuals' beliefs about the external world to be the primitives of the model. Nevertheless, the case when a common prior exist is just a special case of our model: the seller and the buyer shares a common prior distribution $\mu(\omega, \theta)$, where $\omega$ is the state of the world and $\theta$ represents the buyer's private observation, or the buyer's type. Therefore when the buyer's type is realized to $\theta$, he will believe that $\omega$ follows distribution $\mu(\omega | \theta)$. This is a special case of our model, in which $\theta = \mu(\omega |\theta)$.}
 We also assume that the seller will choose and commit to an advertising rule before observing the realization of $\omega$. The timing is as follows
\begin{enumerate}
	\item[(0)] The seller chooses an advertising rule based on $\mu(\omega), \mu(\theta|\omega), u(\omega, a)$ and then posts the advertising rule.
	\item On each day, a new state of the world $\omega$ is drawn from $\mu(\omega)$, independent from what has been observed in the previous rounds. Only the seller observes $\omega$.  Some buyers come, with types $\theta \sim \mu(\theta|\omega)$ independently.
	\item The seller sends a signal $s$ according to the posted advertising rule and set price $p_s$ for the full revelation of $\omega$.
	\item The buyers decide to purchase or not.
\end{enumerate}
The seller's expected revenue per buyer will then equal 
\begin{align*}
	 \sum_{\omega \in \Omega} \mu(\omega) \sum_{\theta \in \Theta} \mu(\theta|\omega) \sum_{s\in S} \pi(\omega, s)\cdot  p_s \cdot \bm{1}(C(\eta^s(\theta)) \ge p_s),
\end{align*}
where $\eta^s(\theta)$ is the type-$\theta$ buyer's posterior belief when receiving $s$, and $C(\cdot)$ is the cost of uncertainty function.
The seller's goal is to find an advertising rule that maximizes his expected revenue (per buyer).

\section{Single Buyer Type}
We start with the case when the seller targets a single buyer type $\theta$, which applies to, for example, the case when the majority of the buyers share a common belief. This simplification also allows us to understand the hardness of the problem and gain some insight into the problem. 
In this case, our problem basically becomes a bayesian persuasion problem with heterogeneous priors~\citep{ALONSO2016672}. The optimal advertising rule can be represented as finding the concave closure of a function that is the product of the \emph{likelihood ratio}~\citep{ALONSO2016672} and the \emph{cost of uncertainty}~\citep{frankel2019quantifying}. We prove that finding the concave closure of this function is NP-hard in general. However, inspired by the observations from~\citep{rayo2010optimal}, we are able to drastically reduce the design space and formulate a finite-size convex program that computes the optimal advertising rule. When the size of $\Omega$ or the size of $A$ is a constant, the convex program  will have a polynomial size. Our convex program also shows that there exists an optimal mechanism that reveals $\le 2|A|$ possible realizations of the state $\omega$ at the beginning. We also apply our observations to the case when $\omega$ is binary and give a characterization of the optimal mechanism.

 
\subsection{Concave Closure Formulation} \label{sec:single_ccf}
When the designer targets a single buyer type, we can apply the characterization in~\citep{ALONSO2016672} to show the follows. 
\begin{definition}[Likelihood ratio~\citep{ALONSO2016672}] \label{def:likelihood}
	For a buyer with prior belief $\theta \in \Delta \Omega$, the \emph{likelihood ratio function} is defined over all of his possible posteriors $\eta$ as follows,
	$$
	R(\eta) = \sum_{\omega} \frac{\mu_\omega\cdot \eta_\omega}{\theta_\omega}.
	$$
\end{definition}
Here the likelihood ratio is well-defined because we assume that $\theta$ has a full support, i.e., $\theta_\omega>0$ for all $\omega$.
\begin{proposition} \label{thm:concave_f}
	When the designer targets a single buyer type $\theta$, her optimal expected revenue (per buyer) is equal to the value of the concave closure of 
	$$f(\eta) = R(\eta)\cdot C(\eta)$$
	at point $\theta$,
	where $R(\eta)$ is the likelihood ratio function (Definition~\ref{def:likelihood}) and $C(\eta)$ is the cost of uncertainty function (Definition~\ref{def:cost_of_unc}). More formally, finding the optimal advertising rule is equivalent to solving the following optimization problem 
\begin{align}
\max_{S,\bm{\eta},\bm{\phi}} &\quad \sum_{s\in S} \phi_s \cdot R(\eta^s) C(\eta^s) \label{prog:concave}\\
\text{s.t.} & \quad \sum_{s\in S} \phi_s \cdot \eta^s = \theta \notag \\
& \quad \phi_s \ge 0, \ \eta^s \in \Delta \Omega, \ \forall s. \notag
\end{align}

\end{proposition}

For completeness, we give the full proof of Proposition~\ref{thm:concave_f} in Appendix~\ref{app:proof_concave}.

The optimization problem basically tries to find a valid set of posteriors with associated probabilities that will maximize the designer's expected revenue. The set $S$ is just the set of signals. The variable $\eta^s$ represents the buyer's posterior after receiving $s$, assuming that his prior is $\theta$.  The variable $\phi_s$ represents the probability of receiving $s$ from the buyer's point of view. If we solve the optimal solution $S, \bm{\eta}, \bm{\phi}$ of~\eqref{prog:concave}, the optimal advertising rule $\langle S, \pi, \{p_s : s \in S\} \rangle$ represented by the solution has
$$ \pi(\omega, s) = \frac{\phi_s \cdot \eta^s_\omega}{\theta_\omega}, \quad \forall \omega \in \Omega, s\in S,$$
$$ p_s = C(\eta^s), \quad \forall s \in S.$$
So the random mapping $\pi$ is decided by the posteriors $\eta^s$ and their associated probabilities $\phi_s$ (from the buyer's point of view), and the optimal price for each signal $s$ is just the highest price that the type-$\theta$ buyer is willing to pay after seeing the signal, which is just the cost of uncertainty at $\eta^s$. 

Note that when the targeted type $\theta = \mu$, the likelihood ratio is always equal to one,
$$
R(\eta) = \sum_{\omega:\,\theta_\omega>0} \frac{\mu_\omega\cdot \eta_\omega}{\theta_\omega} = 1, \ \text{ for all possible posterior } \eta.
$$ As a result, the optimal objective value is the value of the concave closure of  $f(\eta) = C(\eta)$, which is a concave function. The concave closure of a concave function is just itself. As a result, the optimization problem~\eqref{prog:concave} will have an optimal solution
$$
S = \{s\}, \quad \phi_s = 1, \quad \eta^s = \theta,
$$ which means that it is optimal for the seller to not reveal any partial information, but directly set a price for the full revelation.
\begin{proposition} \label{prop:common_prior}
	When  $\theta = \mu$, one of the optimal advertising rules for the seller is to not reveal any partial information and directly charge a price for the full revelation, i.e.
	$$
	S = \{s\}, \quad \pi(\omega, s) = 1 \quad \forall \omega.
	$$
\end{proposition}
But when the targeted type $\theta \neq \mu$, the function $f(\eta) = R(\eta)\cdot C(\eta)$ is neither concave nor convex in general. It will be possible that revealing partial information gives a higher expected revenue. We give an example of such beneficial partial information disclosure. 
\begin{example} \label{emp:concave_emp}
	 Suppose the state is binary and the two possible states happen with equal probability in the long term 
	 $$
	 \Omega = \{0,1\}, \quad \mu(0) = \mu(1) = 0.5.
	 $$
	 The buyers face a problem of guessing the state of the world, with utility equal to one when guessing correctly, equal to zero otherwise
	 $$
	 a\in\{0,1\}, \quad u(\omega, a) = \bm{1}(\omega = a).
	 $$ 
	 The targeted buyers believe that $\omega = 0$ with probability $0.8$, i.e. 
	 $$\theta(0) = 0.8,\quad \theta(1) = 0.2.$$ So the cost of uncertainty function and the likelihood ratio function are
	 $$C(\eta) = \min\{\eta_0, 1-\eta_0\}, \quad R(\eta) = \frac{0.5}{0.8}\eta_0 + \frac{0.5}{0.2} \eta_1.$$ The function $f(\eta) = R(\eta)C(\eta)$ and its concave closure $\overline{f}$ is plotted in Figure~\ref{fig:concave_emp} as a function of $\eta_0$. The plot shows that the optimal expected revenue is $\overline{f}(0.8) = \frac{5}{16}$. Since $(0.8, \overline{f}(0.8))$ is a convex combination of  $(0.5, f(0.5))$ and $(1, f(1))$, the optimal advertising rule sends two possible signals $S = \{s,t\}$ with $\eta^s = (0.5, 0.5)$ and $\eta^t = (1,0)$, which means that the optimal advertising rule has
	 $$
	 \pi(0, s) = \frac{1}{4}, \quad \pi(0, t) = \frac{3}{4}, \quad \pi(1,s) = 1, \quad \pi(1, t) = 0.
	 $$
\end{example}

\begin{figure}
\centering
\begin{tikzpicture}[line cap=round,line join=round,>=triangle 45,x=4cm,y=3.5cm]
\begin{axis}[
x=4cm, y=3.5cm,
axis lines=middle,
xmin=-0.1,
xmax=1.25,
ymin=-0.18,
ymax=1.1,
xtick={0,1},
ytick={0,1},
xlabel={$\eta_0$},
ylabel={$f(\eta_0)$}]
\clip(-0.16475975811684287,-0.13) rectangle (1.177126377704243,1.1712518807599313);
\draw[line width=1.7pt, dashed, smooth,samples=100,domain=0:0.5] plot(\x,{(\x)*(2.5-15/8*(\x))});
\draw[line width=1.3pt,dashed, smooth,samples=100,domain=0.5:1] plot(\x,{(1-(\x))*(2.5-15/8*(\x))});
\draw[line width=1.5pt,color=gray,smooth,samples=100,domain=0:0.5] plot(\x,{(\x)*(2.5-15/8*(\x))});
\draw [line width=1.5pt, color = gray] (0.5,25/32) -- (1,0);
\draw [line width=1pt, dotted, color = black] (0.8, 0) -- (0.8,1);
\draw [line width=1pt, dotted, color = black] (0.5, 0) -- (0.5,1);
\begin{scriptsize}
\draw[color=black] (0.41961001070846876,1.151772888465754) node {$g$};
\draw [fill=gray] (0.8,5/16) circle (2pt);
\draw[color=black] (0.97,0.35) node {$\overline{f}(0.8)$};
\draw[color=black] (0.65,0.13) node {$f(0.8)$};
\draw[color=black] (0.7139592275982553,1.151772888465754) node {$i$};
\draw [fill=gray] (0.5,0.78125) circle (2pt);
\draw[color=gray] (0.8,-0.09) node {$0.8$};
\draw[color=gray] (0.5,-0.09) node {$0.5$};
\draw [fill=gray] (1,0) circle (2pt);
\draw [fill=gray] (0.8,0.2) circle (2pt);

\end{scriptsize}
\end{axis}
\end{tikzpicture}	
\caption{An example of concave closure representation. The function in dashed black line is $f(\eta) = R(\eta)\cdot C(\eta)$ for Example~\ref{emp:concave_emp} and the function in gray line is its concave closure $\overline{f}$.}
\label{fig:concave_emp}
\end{figure}
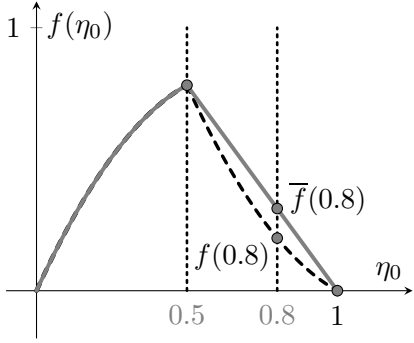

Although we have known that finding the optimal advertising rule is equivalent to finding the concave closure of function $R(\eta) \cdot C(\eta)$, it is still not easy to find the value of the concave closure and the mechanism that will generate the optimal revenue, especially when $\Omega$ is large and the posterior $\eta$ is a high-dimensional vector, in which case we cannot really plot the function $R(\eta) \cdot C(\eta)$. It is also not possible to directly solve the optimization problem~\eqref{prog:concave}, because there are infinitely many choices of $S$ and $\{\eta_s: s\in S\}$.
It turns out that the seller's optimal information advertising problem \eqref{prog:concave} is NP-hard.
\begin{theorem} \label{thm:single_hardness}
	When the seller targets the buyers of a specific type $\theta$, it is NP-hard to find the optimal advertising rule \eqref{prog:concave}.
\end{theorem}
The proof of the theorem can be found in Appendix~\ref{app:single_hardness}.
Despite the hardness result, in the following sections, we show that it is possible to find the optimal advertising rule by a finite-size convex program. When the number of actions $|A|$ is a constant or the number of states $|\Omega|$ is a constant, the optimal advertising rule can be found in polynomial time.

\subsection{Properties of the Optimal Mechanism} \label{sec:single_prop}
In this section, we make some observations about the optimal advertising rule, based on some techniques from~\citep{rayo2010optimal}. These observations will drastically reduce the design space and eventually allow us to solve the optimal advertising rule by a finite-size convex program. 

The first observation (Lemma~\ref{lem:opt_size}) is that there exists an optimal mechanism with $|S|\le n = |\Omega|$. The second and the third observations (Lemma~\ref{lem:opt_decompose} and Lemma~\ref{lem:opt_never})are the necessary conditions for a mechanism to be optimal, by considering the likelihood ratio and the cost of uncertainty generated by the mechanism.

First, having a larger set of signals $S$ may help the seller extract more revenue, but we show that a set of $n$ signals is sufficient for the seller to maximize the expected revenue.
\begin{lemma} \label{lem:opt_size}
There exists an optimal advertising rule $\langle S, \pi \rangle$ with $|S| \le n = |\Omega|$.
\end{lemma}
We give the proof of the lemma in Appendix~\ref{app:single_opt_size}. The idea of the proof is that for any optimal mechanism with $|S|>n$, we can replace one of the signals with a convex combination of the others so that $|S|$ can be decreased by one.

Second, an optimal advertising rule $\langle S, \pi \rangle$ should not have an $\eta^s$ that can be decomposed $\eta^s = \alpha \eta^{(1)} + (1-\alpha) \eta^{(2)}$ to strictly increase the expected revenue
$$
R(\eta^s) C(\eta^s) < \alpha R(\eta^{(1)}) C(\eta^{(1)}) + (1-\alpha) R(\eta^{(2)}) C(\eta^{(2)}).
$$
So $R(\eta^s) C(\eta^s)$ should be locally concave. Let's look at function $R(\eta^s) C(\eta^s)$. The likelihood ratio $R(\eta^s)$ is a linear function of $\eta^s$, and the cost of uncertainty $C(\eta^s) = \min_a C_a(\eta^s)$ is a piece-wise linear function of $\eta^s$. Let 
 \begin{align} \label{eqn:P_a}
\mathcal{P}_a = \{ \eta\in \Delta \Omega: C_a(\eta) \le C_{a'}(\eta), \forall a'\}	
 \end{align}
be the region in which $C(\eta) = C_a(\eta)$, which means that when a buyer's belief falls in $\mathcal{P}_a$, action $a$ will be his best action. We consider the local convexity/concavity of $R(\eta^s) C(\eta^s)$ within $\mathcal{P}_a$. Notice the following fact.
\begin{fact}
Function $g(x,y) = x\cdot y$ with Hessian matrix
\begin{align*}
H = \left[\begin{array}{ll} 0& 1 \\ 1 & 0\end{array} \right]
\end{align*}
is strictly convex along a direction $d = (d_x, d_y)$ with a positive slope $d_x d_y >0$, i.e., for any point $(x_0, y_0)$, function
$$
h(t) = g(x_0 + t d_x,\ y_0 + t d_y)
$$ 
is strictly convex when $d_x d_y >0$.
Because the second directional derivative of $g$ in the direction $d$ at any point $(x_0, y_0)$ is equal to 
\begin{align*}
d^T H d = 2d_x d_y > 0.
\end{align*}
For the same reason, $g(x,y)$ is strictly concave along a direction $d = (d_x, d_y)$ with a negative slope $d_x d_y < 0$.
\end{fact}
Consider $x = R(\eta^s), y = C(\eta^s)$ and $g(x,y) = R(\eta^s)C(\eta^s)$. Then for two-dimensional points $\{(R(\eta), C(\eta)): \eta \in \mathcal{P}_a\}$, we should have the following lemma (illustrated in Figure~\ref{fig:lem_decompose}).  
\begin{lemma} \label{lem:opt_decompose}
	 Let $\langle S, \pi \rangle$ be
	 an optimal advertising rule. Consider a single $s\in S$ with $\phi_s>0$. Let $\eta^s$ be the buyer's
	 posterior when $s$ is sent.
	 Suppose $\eta^s \in \mathcal{P}_a$. Define $\mathcal{Q}_a = \{(R(\eta), C(\eta)): \eta \in \mathcal{P}_a\}$ as the region on xy-plane that represents the likelihood ratio and the cost of uncertainty of the points in $\mathcal{P}_a$. Then the point $(R(\eta^s), C(\eta^s))$ cannot be decomposed along a direction with a positive slope within $\mathcal{Q}_a$, that is, there cannot exist $\eta^{(1)},\eta^{(2)} \in \mathcal{P}_a$ with 
	 $$
	 \eta^s = \alpha \eta^{(1)} + (1-\alpha) \eta^{(2)}, \quad \alpha \in(0,1)
	 $$ 
	 and 
	 $$
	 (R(\eta^s) - R(\eta^{(1)})) (C(\eta^s) - C(\eta^{(1)}))>0.
	 $$
\end{lemma}
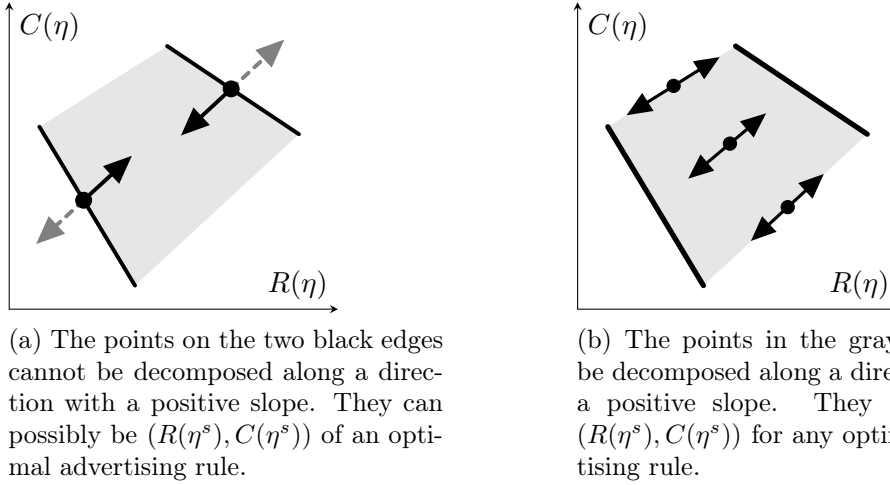
\begin{figure}
\centering

\begin{subfigure}[t]{0.35\textwidth}
\begin{tikzpicture}[line cap=round,line join=round,>=triangle 45,x=0.35cm,y=0.35cm]
\begin{axis}[
x=0.35cm,y=0.35cm,
axis lines=middle,
xmin=2,
xmax=14.489105128421295,
ymin=0,
ymax=11.669696573086298,
xtick={0},
ytick={0},
xlabel={$ R(\eta)$},
ylabel={$C(\eta)$}
]
\clip(-1.9393920529821438,-2.1918479237229467) rectangle (14.489105128421295,11.669696573086298);
\fill[line width=2pt,fill=black,fill opacity=0.1] (3.146917172159998,6.933627877420022) -- (8,10) -- (13.003138520540004,6.650674154500022) -- (6.7781566163,0.8972817884600154) -- cycle;
\draw [line width=1.5pt] (8,10)-- (13.003138520540004,6.650674154500022);
\draw [line width=1.5pt] (6.7781566163,0.8972817884600154)-- (3.146917172159998,6.933627877420022);
\draw [->,line width=1.5pt] (4.827476515624367,4.139970786985745) -- (6.659899440408722,5.829879215634255);
\draw [->, gray, line width=1.5pt,dashed] (4.827476515624367,4.139970786985745) -- (3.0201222620026544,2.473181343654313);
\draw [->,line width=1.5pt] (10.430253449298478,8.373079087170204) -- (8.520940136739583,6.599965021012547);
\draw [->, gray, dashed, line width=1.5pt] (10.430253449298478,8.373079087170204) -- (12.449244954950471,10.248047545503807);
\begin{scriptsize}
\draw [fill=black] (4.827476515624367,4.139970786985745) circle (3pt);
\draw [fill=black] (10.430253449298478,8.373079087170204) circle (3pt);
\end{scriptsize}
\end{axis}
\end{tikzpicture}
\caption{The points on the two black edges cannot be decomposed along a direction with a positive slope.  They can possibly be $(R(\eta^s), C(\eta^s))$ of an optimal advertising rule. }
\end{subfigure}\hspace{0.1\textwidth}
\begin{subfigure}[t]{0.35\textwidth}
\begin{tikzpicture}[line cap=round,line join=round,>=triangle 45,x=0.35cm,y=0.35cm]
\begin{axis}[
x=0.35cm,y=0.35cm,
axis lines=middle,
xmin=2,
xmax=14.232409859961876,
ymin=0,
ymax=11.669696573086298,
xtick={0},
ytick={0},
xlabel={ $ R(\eta)$},
ylabel={$C(\eta)$}
]
\clip(-1.9393920529821496,-2.1918479237229467) rectangle (14.232409859961876,11.669696573086298);
\fill[line width=2pt,fill=black,fill opacity=0.1] (3.146917172159998,6.933627877420022) -- (8,10) -- (13.003138520540004,6.650674154500022) -- (6.7781566163,0.8972817884600154) -- cycle;
\draw [line width=2pt] (8,10)-- (13.003138520540004,6.650674154500022);
\draw [line width=2pt] (6.7781566163,0.8972817884600154)-- (3.146917172159998,6.933627877420022);
\draw [->,line width=1.2pt] (5.645388579288745,8.487440839411118) -- (7.401249869851225,9.584854146012672);
\draw [->,line width=1.2pt] (7.785344527161768,6.29261422620801) -- (9.157111160413706,7.444898198139642);
\draw [->,line width=1.2pt] (5.645388579288745,8.487440839411118) -- (3.8599341209077966,7.3715318029230215);
\draw [->,line width=1.2pt] (7.785344527161768,6.29261422620801) -- (6.199140811010665,4.960203104641079);
\draw [->,line width=1.2pt] (9.980171140364869,3.8783049516845907) -- (11.351937773616806,5.140330254276378);
\draw [->,line width=1.2pt] (9.980171140364869,3.8783049516845907) -- (8.581064270921601,2.5911266317967794);
\begin{scriptsize}
\draw [fill=black] (7.785344527161768,6.29261422620801) circle (2.5pt);
\draw [fill=black] (5.645388579288745,8.487440839411118) circle (2.5pt);
\draw [fill=black] (9.980171140364869,3.8783049516845907) circle (2.5pt);
\end{scriptsize}
\end{axis}
\end{tikzpicture}
\caption{The points in the gray area can be decomposed along a direction with a positive slope.   They cannot be $(R(\eta^s), C(\eta^s))$ for any optimal advertising rule.}
\end{subfigure}
\caption{An illustration of Lemma~\ref{lem:opt_decompose}. We plot the ratio of uncertainty ($x$-coordinate) and the cost of uncertainty ($y$-coordinate) of the points $\eta$ inside polytope $\mathcal{P}_a$. The polygon in the pictures represents the region $\mathcal{Q}_a = \{(R(\eta), C(\eta)): \eta \in \mathcal{P}_a\}$. The two black edges represent the points that can possibly be $(R(\eta^s), C(\eta^s))$ for some $s\in S$ of an optimal advertising rule. The gray area represents the points that cannot be $(R(\eta^s), C(\eta^s))$ for any $s\in S$ of an optimal advertising rule.}
\label{fig:lem_decompose}
\end{figure}


Third, an optimal advertising rule $\langle S, \pi \rangle$ should not have two $\eta^s, \eta^t$ with $s,t\in S$ that can be merged into one signal 
\begin{align} \label{eqn:single_merge}
\eta^{r} = \frac{\phi_s}{\phi_s + \phi_t}\cdot \eta^{s} + \frac{\phi_t}{\phi_s + \phi_t} \cdot \eta^{t}	
\end{align}
 to strictly increase the expected revenue. We prove the follows.

\begin{lemma} \label{lem:opt_never}
The optimal advertising rule $\langle S, \pi \rangle$ should not send two signals $s, t\in S$ with $\phi_s, \phi_t >0$ that have 
\begin{align*}
(R(\eta^s) - R(\eta^t))(C(\eta^s) - C(\eta^t)) < 0.
\end{align*}
\end{lemma}

The omitted full proof can be found in Appendix~\ref{app:single_opt_decompose} and~\ref{app:lem_never}.

\subsection{Optimal Mechanism by Convex Program} \label{sec:single_convex_exp}

With these observations, we reduce~\eqref{prog:concave} to a finite-size convex program. The key idea is to reduce the design space to a finite set by showing that there exists an optimal advertising rule with each $\eta^s$ lying on the segments between the vertices of $\mathcal{P}_a$. Then, fortunately, by defining variables associated with the segments, the expected revenue is convex.

	
As we show in  Figure~\ref{fig:lem_decompose}, point $(R(\eta^s), C(\eta^s))$ for an optimal advertising rule should lie on the boundary of region $\mathcal{Q}_a$. A reasonable conjecture is that $\eta^s$ also lies on the boundary of $\mathcal{P}_a$. We claim that there exists an optimal advertising rule with each $\eta^s$ lying on the segments between the vertices of $\mathcal{P}_a$. 

\paragraph{Vertices of $C(\theta)$.} 
Define $\mathcal{H}_a$ as the set of vertices of the polytope $\mathcal{P}_a$,
$$
\mathcal{H}_a = \big\{ \text{vertices of } \mathcal{P}_a\big\}.
$$
where $\mathcal{P}_a$ is defined in~\eqref{eqn:P_a}. We prove that there exists an optimal advertising rule with each $\eta^s$ lying on the segments between the vertices in $\mathcal{H}_a$ for some $a$.

\begin{lemma} \label{lem:point_pair}
 There exists an optimal advertising rule $\langle S, \pi\rangle$ that has each $\eta^s$ lying on the segments between the vertices in $\mathcal{H}_a$ for some $a$, i.e., for all $s \in S$,
\begin{align*}
&\eta^s = \beta \cdot i + (1-\beta) \cdot j \text{ with } \beta \in [0,1], \\
&i, j \in \mathcal{H}_a \text{ for some } a, 
\end{align*}
 and each $\eta^s = \beta \cdot i + (1-\beta) \cdot j$ must have
$$
(R(i) - R(j))(C(i) - C(j)) \le 0,
$$
and for each pair $i,j$, there is a unique $\eta^s$ that lies on the segment between $i,j$.
\end{lemma}

\begin{figure}
\begin{subfigure}[t]{0.35\textwidth}
\begin{tikzpicture}[line cap=round,line join=round,>=triangle 45,x=0.35cm,y=0.35cm]
\begin{axis}[
x=0.35cm,y=0.35cm,
axis lines=middle,
xmin=2,
xmax=14,
ymin=0,
ymax=11,
xtick={0},
ytick={0},
xlabel={$ R(\eta)$},
ylabel={$C(\eta)$}
]
\clip(-1.2390112325856555,-0.4135013389416602) rectangle (14.99975892730601,12.905939354003209);
\draw [line width=1.5pt, color=gray] (2.902339879923165,1.8397004666703274) -- (9.951751724043808,9.925790523161664);
\draw [->,line width=1.5pt] (5.66926878915395,5.013530686082114) -- (7.2296251254594175,6.8033511894913286);
\draw [->,line width=1.5pt] (5.66926878915395,5.013530686082114) -- (4.082649029366113,3.1935844910313573);
\begin{scriptsize}
\draw [fill=black] (9.951751724043808,9.925790523161664) circle (2pt);
\draw [fill=black] (7.883893351623145,7.553835331267372) circle (2pt);
\draw [fill=black] (2.902339879923165,1.8397004666703274) circle (2pt);
\draw [fill=black] (5.66926878915395,5.013530686082114) circle (3pt);
\draw[color=black] (8,4) node {{$(R(\eta^s), C(\eta^s))$}};
\end{scriptsize}
\end{axis}
\end{tikzpicture}
\caption{$\{(R(i), C(i)): i \in T\}$ lie on a line with a positive slope. The gray segment is the region $\mathcal{Q}_a$.}
\end{subfigure}\hspace{0.1 \textwidth}
\begin{subfigure}[t]{0.35\textwidth}
\begin{tikzpicture}[line cap=round,line join=round,>=triangle 45,x=0.35cm,y=0.35cm]
\begin{axis}[
x=0.35cm,y=0.35cm,
axis lines=middle,
xmin=1,
xmax=13,
ymin=0,
ymax=11,
xtick={0},
ytick={0},
xlabel={$ R(\eta)$},
ylabel={$C(\eta)$}
]
\clip(-1.2390112325856555,-0.4743207028363856) rectangle (14.99975892730601,12.905939354003209);
\fill[line width=2pt,fill=black,fill opacity=0.10000000149011612] (2.167732242368716,5.894053032749113) -- (8.713461588084149,9.680308242525713) -- (10.125285564611007,1.3377120175942223) -- cycle;
\draw [gray, line width=1.5pt,dashed] (6.980768525983005,5.958226849863971) circle (0.68cm);
\draw [->,line width=1.5pt] (6.980768525983005,5.958226849863971) -- (8.26424486828015,6.985007923701692);
\draw [->,line width=1.2pt] (6.980768525983005,5.958226849863971) -- (5.748192729107704,4.972166212363726);
\begin{scriptsize}
\draw [fill=black] (2.167732242368716,5.894053032749113) circle (2pt);
\draw [fill=black] (8.713461588084149,9.680308242525713) circle (2pt);
\draw [fill=black] (10.125285564611007,1.3377120175942223) circle (2pt);
\draw [fill=black] (6.980768525983005,5.958226849863971) circle (3pt);
\draw[color=black] (9.5,5.1) node {$(R(\eta^s), C(\eta^s))$};
\end{scriptsize}
\end{axis}
\end{tikzpicture}
\caption{$\{(R(i), C(i)): i \in T\}$ do not lie on a line. The gray polygon is  $\mathcal{Q}_a$.}
\end{subfigure}

\caption{An illustration for Claim~\ref{clm:nonp_line}. Two cases when $\{(R(i), C(i)): i \in T\}$ do not lie on a line with a nonpositive slope. The big black point represents $(R(\eta^s), C(\eta^s))$, and the small black points are $(R(i), C(i))$ for $i\in T$. $(R(\eta^s), C(\eta^s))$ is a convex combination of the points $\{(R(i), C(i)): i \in T\}$, with positive coefficients. In both of the cases, $(R(\eta^s), C(\eta^s))$ can be decomposed along a direction with a positive slope.}
\label{fig:point_pair}
\end{figure}
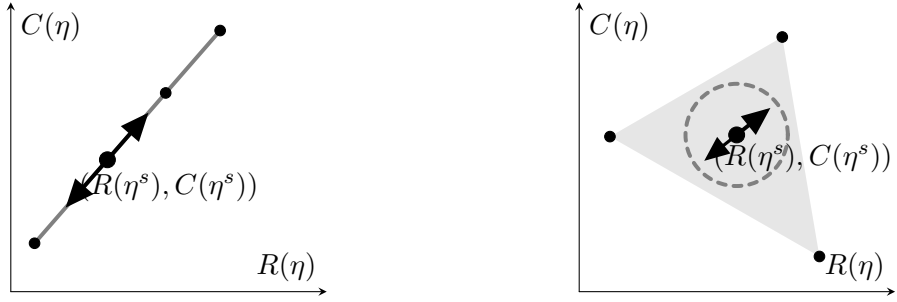

\begin{proof}[Proof Sketch]
Consider an optimal advertising rule $\langle S, \pi\rangle$ and a signal $s\in S$ with $\phi_s>0$. Suppose $\eta^s$ lies in $\mathcal{P}_a$, then $\eta^s$ can be represented as a convex combination of the vertices of the polytope, 
$$
\eta^s = \sum_{i \in \mathcal{H}_a} q_i \cdot i.
$$
Let $T\subseteq \mathcal{H}_a$ be the set of vertices $i$ that has $q_i > 0$. If $|T| \le 2$, then the lemma is proved. Otherwise we claim the follows 
\begin{claim} \label{clm:nonp_line}
The points $\{(R(i), C(i)): i \in T\}$ in two-dimensional space, which represent the likelihood ratio and the cost of uncertainty of $i \in T$, must lie on a line with a nonpositive slope. 
\end{claim}
If the points $\{(R(i), C(i)): i \in T\}$ does not lie on a line with a nonpositive slope, there are two possibilities,
\begin{enumerate}
\item the points lie on a line with a positive slope,
\item the points do not lie on a line.	
\end{enumerate}
In both of the cases, we can decompose $(R(\eta^s), C(\eta^s))$ along a direction $d$ with positive slope (as shown in Figure~\ref{fig:point_pair}).  Therefore according to Lemma~\ref{lem:opt_decompose}, both of the cases cannot be true for an optimal mechanism.
Based on Claim~\ref{clm:nonp_line}, we know that $\{(R(i), C(i)): i \in T\}$ must lie on a line with a nonpositive slope. Then we show that if $\eta^s$ does not lie on a segment between two vertices (in other words $|T|>2$), we can decompose signal $s$ to a bunch of signals that have $\eta$ lying on segments between vertices and have the same likelihood ratio and cost of uncertainty. 
Finally, by Lemma~\ref{lem:opt_never}, we  can have a unique $\eta^s$ lying on the segment between $i,j$ for each pair $i,j$.
The full proof is in Appendix~\ref{app:point_pair}.
\end{proof}

With Lemma~\ref{lem:point_pair}, we are ready to formulate a convex program to compute the optimal advertising rule. Let $\mathcal{G}$ be the set of all possible vertices pairs that can possibly have $\eta^s$ lying between them. 
$$
\mathcal{G} = \left\{ \{i,j\} : i,j \in \mathcal{H}_a \text{ for some } a, (R(i) - R(j))(C(i) - C(j)) \le 0\right\}.
$$
We represent the posterior $\eta^s$ lying on the segment between $i,j$ as $\phi_s \cdot \eta^s =\phi_s (\beta i + (1-\beta) j) = \gamma_{ij}\cdot i + \gamma_{ji} \cdot j$. Then $\phi_s\cdot R(\eta^s) (\eta^s)$ can be represented as
\begin{align}
&\phi_s\cdot R(\eta^s) C(\eta^s) \notag \\
= \ &  (\gamma_{ij} + \gamma_{ji}) \cdot\frac{\gamma_{ij} R(i) + \gamma_{ji} R(j)}{\gamma_{ij} + \gamma_{ji}}
	\cdot\frac{\gamma_{ij} C(i) + \gamma_{ji} C(j)}{\gamma_{ij} + \gamma_{ji}} \notag\\
= \ &  \gamma_{ij} R(i) C(i) + \gamma_{ji} R(j) C(j) - \frac{\gamma_{ij} \gamma_{ji}}{\gamma_{ij} +\gamma_{ji}} (R(i) - R(j))(C(i) - C(j)). \label{eqn:ojb_RC}
\end{align}
For $\{i,j\}$ that has $(R(i) -R(j))(C(i) - C(j))\le 0$, \eqref{eqn:ojb_RC} is a concave function of $(\gamma_{ij}, \gamma_{ji})$, because $\frac{\gamma_{ij} \gamma_{ji}}{\gamma_{ij} +\gamma_{ji}}$ with negative semidefinite Hessian 
\begin{align*}
H = \frac{1}{(\gamma_{ij}+\gamma_{ji})^3} \left[\begin{array}{ll} - \gamma_{ji}^2 & \gamma_{ij} \gamma_{ji} \\ \gamma_{ij} \gamma_{ji} & -\gamma_{ij}^2 \end{array} \right]
\end{align*}
is a concave function. 
So we can rewrite the optimization problem \eqref{prog:opt2} as a convex program as follows.
\begin{theorem} \label{thm:main_convex}
Define $
\mathcal{G} = \left\{ \{i,j\} : i,j \in \mathcal{H}_a \text{ for some } a, (R(i) - R(j))(C(i) - C(j)) \le 0\right\}.
$
The following convex program finds an optimal advertising rule
\begin{align}
    \max & \quad \sum_{\{i,j\} \in \mathcal{G}} \gamma_{ij} R(i) C(i) + \gamma_{ji} R(j) C(j) - \frac{\gamma_{ij} \gamma_{ji}}{\gamma_{ij} +\gamma_{ji}} (R(i) -R(j))(C(i) - C(j)) \label{prog:opt_convex}\\
    \text{s.t.} & \quad \sum_{\{i,j\} \in \mathcal{G}} \gamma_{ij}\cdot i + \gamma_{ji} \cdot j = \theta \notag\\
    & \quad \gamma \ge 0 \notag
\end{align}
\end{theorem}

\subsection{Extreme points as basic feasible solutions} \label{sec:single_convex_poly}

We have known how to find the optimal advertising rule by a convex program given the vertices $\mathcal{H}_a$ of the polytopes $\mathcal{P}_a$. The vertices can be found by finding the \emph{basic feasible solutions} of the linear equations that defines $\mathcal{P}_a$. We show how to find these basic feasible solutions  in Appendix~\ref{app:extreme}.

When $|A|$ is a constant or $|\Omega|$ is a constant, the linear equations that defines $\mathcal{P}_a$ have polynomially many basic feasible solutions. So the convex program \eqref{prog:opt_convex} will have a polynomial size and the optimal advertising rule can be solved in polynomial time.
\begin{theorem} \label{cor:poly}
	When the number of actions $|A|$ is a constant or $|\Omega|$ is a constant, we can find an optimal advertising rule $\langle S, \pi\rangle$ with $|S|\le n$ within polynomial time.
\end{theorem}
\begin{proof}
	We can first find an optimal advertising rule using the convex program in Theorem~\ref{thm:main_convex}, which has a polynomial size according to Lemma~\ref{lem:num_vertices}. Then we can reduce the size of $S$  to at most $n$ by the method in Lemma~\ref{lem:opt_size}.
\end{proof}

As we show in Appendix~\ref{app:extreme}, the basic solutions of the linear equations that define $\mathcal{P}_a$ have at most $|A|$ non-zero entries.  According to Lemma~\ref{lem:point_pair}, there exists an optimal advertising rule that has the buyer's posteriors lying on the segments between the vertices. Therefore, there exists an optimal advertising rule which guarantees that $\eta^s$ has no more than $2 |A|$ non-zero entries for any signal realization $s$. This means that there exists an optimal advertising rule which will allow the buyer to reduce the range of $\omega$ to a set of size $\le 2|A|$ after seeing the partial information provided by the seller.
\begin{theorem}
	There exists an optimal advertising rule $\langle S, \pi \rangle$ with $|S| \le n$ that guarantees for all $s \in S$, the buyer's posterior $\eta^s$ has no more than $2 |A|$ non-zero entries,
	$$
	\Vert \eta^s \Vert_0 \le 2|A|.
	$$
\end{theorem}

Furthermore, when $\theta = \mu$, the number of possibilities can be further reduced to $|A|$.
\begin{proposition} \label{prop:common_reveal}
When $\theta = \mu$, there exists an optimal advertising rule $\langle S, \pi \rangle$ with $|S| \le n$, that reveals $\le |A|$ possibilities of the realized state of the world $\omega$ to the buyer before selling the (remaining) information, i.e., for all $s\in S$,
$$
	\Vert \eta^s \Vert_0 \le |A|.
	$$
\end{proposition}
Recall that in Proposition~\ref{prop:common_prior}, we show that when $\theta = \mu$,  it is optimal for the seller to not reveal anything and directly charge the buyer his expected gain. Here Proposition~\ref{prop:common_reveal} implies that there is another optimal strategy of revealing some information to buyers that achieves the same revenue as revealing nothing and helps the buyer to narrow down the set of possible states. This is because $\theta$ can be decomposed into points in $\mathcal{H}_a$ for some $a\in A$ without changing the expected revenue. 
The proofs can be found in Appendix~\ref{app:extreme}.

\subsection{Optimal Mechanism for Binary State} \label{sec:single_binary}
In this section, we use the results in the previous sections to give some characterizations of the optimal advertising rule for the case when the state of the world is binary, i.e. $|\Omega|=2$. First, according to Lemma~\ref{lem:opt_size}, there exists an optimal mechanism that only sends two possible signals to the buyer. In addition, the cost of uncertainty $C(\eta) = C((\eta_1, \eta_2)) =  C((\eta_1, 1-\eta_1))$ can be represented as a function of $\eta_1$, 
\begin{align*}
C(\eta_1) & = \min_a C_a(\eta)\\ 
& = \min_{a} \eta_1 \cdot \Delta u(1, a) + (1 - \eta_1) \Delta u(2,a)\\
& = \min_a \Delta u(2,a) + \eta_1 (\Delta u(1, a) - \Delta u(2,a))
\end{align*}
where $\Delta u(\omega, a) = \max_{a'} u(\omega, a') - u(\omega, a)$. So $C(\eta_1)$ is the minimum of $|A|$ linear functions of $\eta_1$, as shown in Figure~\ref{fig:single_binary}. We define vertices of $C(\eta_1)$ to be the turning points of function $C(\eta_1)$, and let $\eta^*_1 = \arg \max_{\eta_1} C(\eta_1)$ be the maximum point of $C(\eta_1)$. Then we claim the follows.

\begin{figure}
\centering
\begin{subfigure}[t]{0.4\textwidth}
\begin{tikzpicture}[line cap=round,line join=round,>=triangle 45,x=0.3cm,y=0.35cm]
\begin{axis}[
x=0.3cm,y=0.35cm,
axis lines=middle,
xmin=-0.8,
xmax=15.5,
ymin=-1.5,
ymax=9.8,
xtick={0},
ytick={0},
xlabel={$\eta_1$},
ylabel={$C(\eta_1)$}]
\clip(-1.2390112325856555,-3.028733986414866) rectangle (14.99975892730601,10.108248614845827);
\draw [line width=2pt, dashed, color=gray] (0,0)-- (2.592608692782041,4.695325228215276);
\draw [line width=2pt, dashed, color=gray] (2.592608692782041,4.695325228215276)-- (7.032422257096991,6.763183600635941);
\draw [line width=2pt] (7.032422257096991,6.763183600635941)-- (11.47223582141194,4.026312225373297);
\draw [line width=2pt] (11.47223582141194,4.026312225373297)-- (13.479274829937877,0);
\draw [line width=1.5pt, color=gray] (7.032422257096991,-3.028733986414866) -- (7.032422257096991,10.108248614845827);
\begin{scriptsize}
\draw [fill=black] (0,0) circle (3pt);
\draw[color=black] (0.1,-0.8) node {\Large $0$};
\draw [fill=black] (2.592608692782041,4.695325228215276) circle (3pt);
\draw [fill=black] (7.032422257096991,6.763183600635941) circle (3pt);
\draw[color=black] (13.357636102148426,-0.7) node {\Large $1$};
\draw[color=black] (7.82307398772842,-0.7) node {\Large $\eta^*_1$};
\draw[color=black] (4,7.2) node {\Large $\eta^s_1$};
\draw[color=black] (10,7.2) node {\Large $\eta^t_1$};
\end{scriptsize}
\end{axis}
\end{tikzpicture}
\caption{An example with $4$ actions.}
\end{subfigure}\hspace{0.08\textwidth}
~
\begin{subfigure}[t]{0.35\textwidth}
\begin{tikzpicture}[line cap=round,line join=round,>=triangle 45,x=0.3cm,y=0.35cm]
\begin{axis}[
x=0.3cm,
y=0.35cm,
axis lines=middle,
xmin=-0.8,
xmax=15,
ymin=-1.5,
ymax=9.8,
xtick={0},
ytick={0},
xlabel={$\eta_1$},
ylabel={$C(\eta_1)$}]
\clip(-1.2390112325856555,-3.028733986414868) rectangle (14.99975892730601,10.108248614845827);
\draw [line width=1.5pt,color=gray] (5.08620261246578,-3.028733986414868) -- (5.08620261246578,10.108248614845827);
\draw [line width=1.5pt,dashed, color=gray] (0,0.012234208321404033)-- (5.08620261246578,6.641544872846478);
\draw [line width=2pt] (5.08620261246578,6.641544872846478)-- (12.871081190990623,0);
\begin{scriptsize}
\draw [fill=black] (0,0.012234208321404033) circle (3pt);
\draw [fill=black] (5.08620261246578,6.641544872846478) circle (3pt);
\draw[color=black] (0.1,-0.8) node {\Large $0$};
\draw[color=black] (5.876854343097209,-0.7) node {\Large $\eta^*_1$};
\draw[color=black] (12.749442463201174,-0.7) node {\Large $1$};
\draw[color=black] (3,6.2) node {\Large $\eta^s_1$};
\draw[color=black] (8,6.2) node {\Large $\eta^t_1$};
\end{scriptsize}
\end{axis}
\end{tikzpicture}
\caption{An example with binary action.}
\end{subfigure}
\caption{Characterization of the optimal advertising rule when the state of the world is binary. We plot the function $C(\eta_1)$, which is the minimum of $|A|$ linear functions of $\eta_1$. The optimal advertising rule must have $\eta^s_1$ being one of the vertices on one side of $\eta^*_1$, plotted as black dots in the pictures. And $\eta^t_1$ must lie on the other side of $\eta^*_1$, plotted as black segments in the pictures.}
\label{fig:single_binary}
\end{figure}
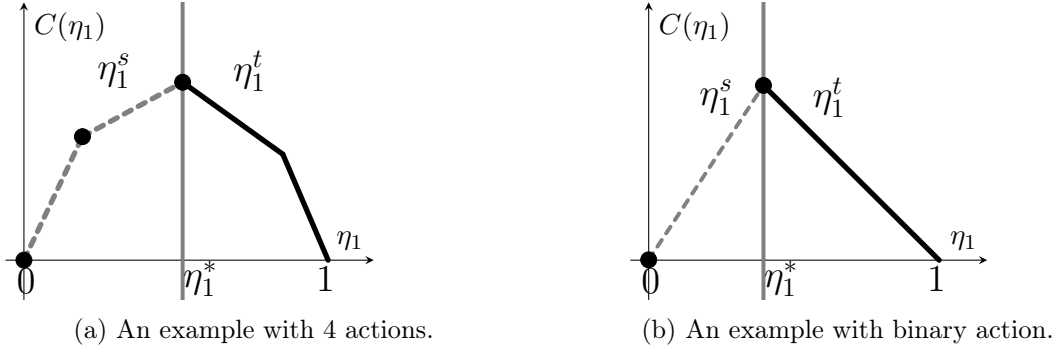

\begin{theorem}\label{thm:single_binary}
When $|\Omega| = 2$, 
there exists an optimal advertising rule $\langle S, \pi\rangle$ with $S= \{ s, t\}$. The optimal advertising rule has $\eta^s_1$ being a vertex of $C(\eta_1)$ on one side of $\eta^*_1$, and $\eta^t_1$ lying on the other side of $\eta^*_1$, as illustrated in Figure~\ref{fig:single_binary}. The optimal advertising rule can be solved in $O(|A|^2)$ time.
\end{theorem}
\begin{proof}[Proof sketch]
First by Lemma~\ref{lem:opt_size}, there exists an optimal advertising rule that has $|S| = 2$. Let $S = \{s, t\}$.
Without loss of generality assume $R((1,0) > R((0,1))$. Then $R(\eta_1)$ is an increasing function of $\eta_1$.  By Lemma~\ref{lem:opt_decompose}, the optimal advertising rule should not have $\eta^s_1$ or $\eta^t_1$ lying on the left of $\eta_1^*$ and not being a vertex. By Lemma~\ref{lem:opt_never}, the optimal advertising rule should not have both $\eta^s_1$ and $\eta^t_1$ on the right of $\eta_1^*$.  Therefore we must have one of $\eta^s_1$ and $\eta^t_1$ being a vertex on the left of $\eta^*_1$, and the other one on the right of $\eta^*_1$.

The optimal advertising rule can be solved in $O(|A|^2)$ time by enumerating all possibilities of $\eta^s_1$ and $\eta^t_1$, and then compute a tangent line of $R(\eta_1) C(\eta_1)$. Details can be found in Appendix~\ref{app:single_binary}.
\end{proof}

\begin{corollary}
When $|\Omega| = 2$ and $|A| = 2$, there exists an optimal advertising rule $\langle S, \pi \rangle$ that sends two possible signals $S = \{ s, t \}$. When the buyer receives signal $s$, he knows the value of $\omega$ exactly, 
\begin{align*}
\pi(1, 	s) = 0, \pi(2, s) \in [0,1] \quad \text{ or } \quad \pi(2, 	s) = 0, \pi(1, s) \in [0,1], 
\end{align*}
and the seller charges nothing,
$$
p_s = 0.
$$
\end{corollary}

\section{General Problem}
We then move to the general problem when the seller faces buyers of different types (i.e. $|\Theta|>1$) and the types are drawn from distribution $\mu(\theta|\omega)$. It turns out that the general problem  is not only NP-hard to solve, but also NP-hard to approximate within a constant factor. We thus turn to some special cases of the problem and give a linear program approximation algorithm that finds an $\varepsilon$-suboptimal mechanism when it is not too hard to predict the possible type of buyers who will make the purchase. This algorithm can also be used when a single type of buyers is targeted and find a $\varepsilon$-suboptimal mechanism.

\subsection{Hardness of General Problem}

In the general problem of optimal  advertising, the seller first advertises the information by sending a signal. Then for each possible realization of the signal, a buyer's valuation of the remaining information will follow a distribution that can be inferred by the seller. Based on this distribution, the seller chooses the best price that will maximize her expected revenue, that is, the price times the probability that the buyer will make the purchase. 
The problem is hard in general. 
\begin{theorem} \label{thm:general_hardness}
Given the support of the state of the world $\Omega = \{1, \dots, n\}$, the support of the buyers' personal beliefs of the state of the world $\Theta \subseteq \Delta \Omega$ and a joint distribution over the two $\mu(\omega, \theta)$, as well as the buyers' utility function $u(\omega, a)$ for $\omega \in \Omega, a\in A$, it is NP-hard to find a constant-factor approximation of the optimal advertising rule that maximizes the seller's revenue in expectation.
\end{theorem}
 
\subsection{Approximation for Special Cases}

Due to the hardness of the general problem, we investigate some special cases. 
We show that for some special cases when it is not too hard to predict the possible types of buyers who will make the purchase, it is possible to find an advertising rule  with revenue arbitrarily close to the optimal mechanism within poly$(1/\varepsilon, |A|, |\Omega|, |\Theta|)$ running time, where $\varepsilon$ is the upper bound of the difference between our mechanism and the optimal mechanism. 

Suppose now the seller sends a signal $s$ and charges a price $p_s$. Let  $\Lambda(s, p_s) \subseteq \Theta$ be the set of buyer types that would pay for the full revelation of $\omega$, i.e.,
\begin{align} \label{def:lambda_s}
\Lambda(s, p_s) = &\{\theta: C(\eta^s(\theta)) \ge p_s\} 
\end{align}
Let $\mathbf{\Lambda}$ be the set of all possible $\Lambda(s, p_s)$,
\begin{align*}
\mathbf{\Lambda} = \left\{\Lambda(s, p_s): s \text{ is a signal sent by a signaling scheme } \pi, p_s \in \mathbb{R}\right\} 
\end{align*}
For some special cases, the number of possible $\Lambda(s, p_s)$, i.e. $|\mathbf{\Lambda}|$ will not be too large.
\begin{itemize}
	\item When the type space is relatively small, $|\Theta| = O(\log N)$ where $N = \max\{ |A|, |\Omega|\}$. Then the number of all possible subsets of $\Theta$ is $O(N)$.
	\item When there is a binary state of the world, i.e. $|\Omega| = 2$.  In this case,  $|\mathbf{\Lambda}|$ is no more than $|\Theta|^2$. This is because $\Lambda(s, p_s)$ must be  a convex set, due to its definition~\eqref{def:lambda_s} and the concavity of the cost of uncertainty function $C(\cdot)$. Therefore, denoting $\theta = (\theta_1, \theta_2)$, the types $\theta$ in $\Lambda(s, p_s)$ must have $\theta_1$ lying in an interval $[L,R]$.  Since the type space is discrete, we only need to consider the intervals with endpoints $\{\theta_1: \theta \in \Theta\}$ to include all $\Lambda(s, p_s)$. The number of such intervals is no more than $|\Theta|^2$. We give the full proof in Appendix~\ref{app:eff_bi}.
\end{itemize}
We show that if there are only polynomially many possible $\Lambda(s, p_s)$, i.e., $|\mathbf{\Lambda}| = \text{poly} (|A|, |\Omega|, |\Theta|)$, there exists an approximation algorithm that can approximate the optimal revenue arbitrarily close.
\begin{theorem}
	Given $\Omega=\{1, \dots, n\}$, $\Theta \subseteq \Delta \Omega$ and a joint distribution over the two $\mu(\omega, \theta)$, as well as the buyers' utility function $u(\omega, a)$ for $\omega \in \Omega, a\in A$.
	If there are only polynomially many possible $\Lambda(s, p_s)$, i.e., $|\mathbf{\Lambda}| = \text{poly} (|A|, |\Omega|, |\Theta|)$, then there exists an algorithm that for any $\varepsilon \in (0,1)$, finds an advertising rule that achieves expected revenue at least $OPT - \varepsilon$ within poly$(1/\varepsilon, |A|, |\Omega|, |\Theta|)$ running time, where $OPT$ is the expected revenue of the optimal advertising rule.
\end{theorem}
To prove the theorem, we first show that there exists an optimal advertising rule that has each signal $s$ mapping to a unique $(p_s, \Lambda(s, p_s))$.
\begin{lemma} \label{lem:general_eff}
There exists an optimal advertising rule $\langle S, \pi, \{p_s : s \in S\} \rangle$ that satisfy the follows: for any two different signals $s, t \in S$, either $p_s \neq p_t$ or  $\Lambda(s, p_s)) \neq \Lambda(t, p_t))$. In other words, each $s\in S$ has a unique $(p_s, \Lambda(s, p_s))$.
\end{lemma}
We give the proof of the lemma in Appendix~\ref{app:general_eff}. The idea is that if there are two signals with the same $(p_s, \Lambda(s, p_s))$, we can merge them into one.

Now we formulate an LP to compute an approximately optimal advertising rule. 
Since we assume $u(\omega,a) \in [0,1]$, the prices charged by the optimal mechanism must lie in $[0,1]$. Then we can approximate the prices by choosing  $p_s$ from a finite set $P= \{0, \varepsilon, 2\varepsilon, \dots, 1 \}$ with size $\le \lfloor 1/\varepsilon \rfloor + 1$. Since the optimal mechanism only needs one signal for each price and each possible $\Lambda(s, p_s)$,  we assign one signal $s_{p, \Lambda}$ for each  pair of $p \in P$ and $\Lambda \in \mathbf{\Lambda}$, so that $S = \{s_{p, \Lambda}: p \in P, \Lambda \in \mathbf{\Lambda}\}$.
Then let the variables of the LP  be the probability transition function of the signaling scheme $\pi(\omega, s_{p, \Lambda})$ for all $\omega \in \Omega$ and $s_{p, \Lambda} \in S$. We add the constraints so that when $s_{p, \Lambda}$ is sent and the price is set to $p$, all the types in $\Lambda$ will be willing to make the purchase,
$$
C\big(\eta^{s_{p, \Lambda}}(\theta)\big) \ge p, \quad \forall p\in P, \Lambda \in \mathbf{\Lambda}, \theta\in \Lambda.
$$
By the definition of the cost of uncertainty function (Definition~\ref{def:cost_of_unc}) and the posterior~\eqref{eqn:posterior_eta},  
this can be equivalently represented by linear constraints
\begin{align} \label{eqn:gnrl:constraint}
\sum_{\omega \in \Omega} \theta_\omega \cdot \pi(\omega, s_{p, \Lambda})(u^*(\omega) - u(\omega, a)) \ge p  \sum_\omega \theta_\omega \cdot \pi(\omega, s_{p, \Lambda}), \quad \forall p\in P, \Lambda \in \mathbf{\Lambda}, \theta\in \Lambda.
\end{align}
  Then the expected revenue is a linear function of the variables
\begin{align} \label{eqn:gnrl:obj}
\sum_{p\in P, \Lambda \in \mathbf{\Lambda}} p \sum_{\omega, \theta \in \Lambda} \mu(\omega, \theta) \pi(\omega, s_{p, \Lambda}).
\end{align}
Finally we add constraints so that $\pi(\cdot)$ is a valid signaling scheme
\begin{align} \label{eqn:gnrl:con_sig}
\sum_{s \in S} \pi(\omega, s) = 1, \ \forall \omega.	\qquad \pi(\omega, s) \ge 0, \ \forall \omega, s.
\end{align}
The LP with objective~\eqref{eqn:gnrl:obj} and constraints~\eqref{eqn:gnrl:constraint} and~\eqref{eqn:gnrl:con_sig} computes an advertising rule with expected revenue at least $OPT-\varepsilon$ because by rounding the prices of an optimal advertising rule down to its closest price in $P$, we get a feasible solution of the LP, and this will not decrease the expected revenue by more than $\varepsilon$.

\section{Conclusion and Future Work}
In this work, we study the problem of optimal advertising for information products. We prove the hardness of the problem and present positive results in both the simple setting and the general setting. There are many directions left open for future work.
\begin{itemize}
\item The most 	appealing open problem would probably be how to get around the strong impossibility results. In this work, we have considered general decision problems and arbitrary distributions. Can one come up with some special but non-trivial utility functions or distributions so that the problem will be tractable? What are the necessary assumptions we need to add for the problem to be easy? 
\item For our model of general decision problems and arbitrary distributions, there are also some interesting open questions. What are the best approximation algorithms we can find for the general problem? In particular, our hardness result assumes that there is no common prior between the seller and the buyer. Will the problem still be hard when the buyer and the seller share a common prior? Does the common prior assumption matter?
\item It would also be interesting to extend our model to other problems. For example, in this work, the seller only cares about the revenue. What if the seller also cares about the buyer's action? We have studied the design of the optimal advertising rule when it is decided by the seller. What if the advertisement is instead provided by a third-party agent? What would be the best advertising strategy of this third-party advertisement provider, and how would it affect the social welfare? We're also not considering the advertising cost. A natural extension is to incorporate the cost into the model.
\end{itemize}

%
%
%
%
%

\bibliographystyle{plainnat}
\bibliography{ref}

\newpage
\appendix
\section{Single Buyer Type}
\subsection{Proof of Proposition~\ref{thm:concave_f}} \label{app:proof_concave}

We first formulate our optimal advertising problem as an optimization problem. When the seller targets a single buyer type $\theta$, 
 the optimal advertising rule can fully extract the expected surplus from the buyers of that type after sending the signal, i.e., an optimal advertising rule $\langle S, \pi, \{p_s : s \in S\} \rangle$ must have $p_s = C(\eta^s(\theta))$ for the targeted type $\theta$. To simplify the notation, we denote by $\langle S, \pi \rangle$ an advertising rule, and use $\eta^s$ to represent the posterior $\eta^s(\theta)$.
\paragraph{Optimal mechanism formulation.}
Recall that when the signal is realized to $s$, the posterior belief of the buyer is
\begin{align}\label{eqn:posterior}
\eta^s = \frac{\big( \theta_1 \pi(1, s), \dots, \theta_n \pi(n, s)\big)}{\sum_{\omega=1}^n \theta_\omega \pi(\omega,s)}.
\end{align}
The optimal mechanism charges the buyer his cost of uncertainty $p_s = C(\eta^s)$ when $s$ is realized. Define $\phi_\mu(s) = \sum_\omega \mu_\omega \pi(\omega, s)$ to be the probability of sending $s$. Then the seller's expected revenue is equal to $\sum_{s\in S} \phi_\mu(s) \cdot C(\eta^s)$ and the seller's optimization problem can be formulated as
\begin{align} 
\max_{S, \pi} &\quad \sum_{s\in S} \phi_\mu(s) \cdot C(\eta^s) \label{prog:opt1}\\
\text{s.t.} & \quad \sum_{s\in S} \pi(\omega, s) = 1, \quad \forall \omega \notag\\
& \quad \pi(\omega, s) \ge 0, \quad \forall \omega, s. \notag
\end{align}
Observe that the probability of sending a signal $\phi_\mu(s) = \sum_\omega \mu_\omega \pi(\omega, s)$ depends on the true underlying distribution $\mu$ but not $\theta$,  while $C(\eta^s)$ depends on the buyer's belief $\theta$. We show that we can rewrite  $\phi_\mu(s)$ as well as the constraints as functions of $\theta$, so that the whole optimization can be viewed as finding the concave closure of a function $f(x)$ at point $\theta$.
\paragraph{Concave closure representation.} 
Let $\phi_\theta(s) = \sum_\omega \theta_\omega \pi(\omega, s)$ be the probability of receiving $s$ based on the buyer's personal belief. The ratio $\phi_\mu(s)/\phi_\theta(s)$ can be determined as long as we know the posterior $\eta^s$, i.e., we can define  the ratio $\phi_\mu(s)/\phi_\theta(s)$ as a function of $\eta^s$,
\begin{align*}
    R(\eta^s) = \frac{\phi_\mu(s)}{\phi_\theta(s)} = \frac{\sum_\omega \mu_\omega \pi(\omega, s)}{\sum_\omega \theta_\omega \pi(\omega, s)}  = \sum_\omega \mu_\omega \cdot \frac{\eta^s_\omega}{\theta_\omega}.
\end{align*}
The last equality is because $\pi(\omega, s)/\left(\sum_\omega \theta_\omega \pi(\omega, s)\right) = \eta^s_\omega/\theta_\omega$ according to~\eqref{eqn:posterior}. We call $R(\eta^s)$ \emph{the likelihood ratio function}. Note that $R(\eta^s)$ is a linear function of $\eta^s$ with coefficients $\mu_\omega/\theta_\omega$.
Then the seller's expected revenue can be represented as the expected product of the likelihood ratio and the cost of uncertainty,
$$
\sum_{s\in S} \phi_\theta(s) \cdot R(\eta^s) C(\eta^s).
$$
According to~\eqref{eqn:posterior} and $\phi_\theta(s) = \sum_\omega \theta_\omega \pi(\omega, s)$, we have $\phi_\theta(s) \cdot \eta^s = \big( \theta_1 \pi(1, s), \dots, \theta_n \pi(n, s)\big)$. So the constraints $\sum_{s\in S} \pi(\omega, s) = 1, \pi(\omega, s) \ge 0$ can be equivalently written as
$$
\sum_{s\in S} \phi_\theta(s) \cdot \eta^s = \theta, \qquad \phi_\theta(s) \ge 0, \ \eta^s \in \Delta \Omega, \ \forall s.
$$
Therefore the seller's problem~\eqref{prog:opt1} can be equivalently represented as 
\begin{align} 
\max_{\eta,\,\phi_\theta} &\quad \sum_{s\in S} \phi_\theta(s) \cdot R(\eta^s) C(\eta^s) \label{prog:opt2}\\
\text{s.t.} & \quad \sum_{s\in S} \phi_\theta(s) \cdot \eta^s = \theta \notag \\
& \quad \phi_\theta(s) \ge 0, \ \eta^s \in \Delta \Omega, \ \forall s \notag
\end{align}
Observe that the optimal objective value of~\eqref{prog:opt2} is just the value of the \emph{concave closure} of the product of the likelihood ratio and the cost of uncertainty $f(x) = R(x)\cdot C(x)$ at position $x = \theta$.

\subsection{Proof of Theorem~\ref{thm:single_hardness}} \label{app:single_hardness}
As shown in Section~\ref{sec:single_ccf}, the seller's problem
\begin{align*} 
\max_{\eta,\,\phi_\theta} &\quad \sum_{s\in S} \phi_\theta(s) \cdot R(\eta^s) C(\eta^s) \\
\text{s.t.} & \quad \sum_{s\in S} \phi_\theta(s) \cdot \eta^s = \theta \\
& \quad \phi_\theta(s) \ge 0, \ \eta^s \in \Delta \Omega, \ \forall s 
\end{align*}
is equivalent to find the concave closure of $f(x) = R(x)\cdot C(x)$ at a point $\theta$, where $R(x)$ is the likelihood ratio, and $C(x) = \min_a C_a(y)$ is the cost of uncertainty. The concave closure of $f(x)$, denoted by $\overline{f}(x)$, is equal to
\begin{align} \label{eqn:CC_dual}
    \overline{f}(x) = \min_{\alpha, \beta} \{ \alpha^T x + \beta~|~\alpha^T y + \beta \ge f(y), \forall y\in \Delta_n\}.
\end{align}
For simplicity, in this section we allow $C_a(y) = c_a^T y$ to have negative coefficients (by definition, $C_a(y)$ should always has non-negative coefficients). This is without loss of generality because we can always equivalently consider $\widetilde{C}_a(y) = \sum_{i=1}^n y_i \big(\max_{a'} C_{a'}(e_i) - C_{a}(e_i) \big)$, which is a valid cost of uncertainty function with non-negative coefficients. 

We will introduce a new problem that is closely related to~\eqref{eqn:CC_dual}. We will prove the hardness of this new problem and then use it to prove the hardness of \eqref{eqn:CC_dual}.

Since the feasible solution $(\alpha, \beta)$ of \eqref{eqn:CC_dual} forms a convex set, the ellipsoid method can be applied to solve~\eqref{eqn:CC_dual} if there is a cutting-plane oracle that, for any point $(\alpha, \beta)$, returns a $y$ such that $\alpha^T y + \beta < f(y)$ if there exists one and returns ``feasible'' if $\alpha^T y + \beta \ge f(y), \forall y$. To have such a cutting-plane oracle, it suffices to solve
\begin{align*}
    \max_{y\in \Delta_n} f(y) - (\alpha^T y + \beta) = R(y)\cdot C(y) -  \alpha^T y,
\end{align*}
where $R(x) =  \sum_{i=1}^n x_i \cdot \frac{\mu_i}{\theta_i}$ is a linear function and $C(x) = \min_a C_a(x)$ in which $C_a(x)$ is a linear function of $x$. This can be solved by solving 
\begin{align} \label{eqn:dual_oracle}
    \max_{y\in \Delta_n} \{ R(y)\cdot C_a(y) -  \alpha^T y ~|~ C_a(y) \le C_{a'}(y), \forall a'\}.
\end{align}

We first show that~\eqref{eqn:dual_oracle} is hard to solve for specific $R(y), C_a(y), \alpha$.
\begin{lemma}
There exist fixed $R(y) = r^T y$, $C_a(y) = c_a^T y$, fixed $\alpha$, and partially fixed $C(y)$, such that deciding whether the solution of~\eqref{eqn:dual_oracle} is greater or equal to $0$ is NP-complete, and the maximum must be achieved at point $y$ with
$$
y_j \in \{0, 1/n\} \text{ for } 1\le j \le n-1, \quad y_n = 1 - \sum_{j=1}^{n-1} y_j.
$$
\end{lemma}
\begin{proof}
We use reduction from the following problem, which is proved to be NP-complete in~\cite{pardalos1991quadratic}.
\begin{description}
\item[\cite{pardalos1991quadratic}.] There exist fixed non-negative vectors $\gamma, \beta$ with length $n$, and partially fixed $A,b$, so that it is NP-complete to decide whether 
\begin{align} 
\max \quad & (\beta^T y)^2 - \gamma^T y \label{eqn:single_hard_eq1}\\
\text{s.t.} \quad & Ay \le b \notag \\
& 0 \le y_i \le 1, \quad \forall i\in[n]. \notag
\end{align}
And the maximum of \eqref{eqn:single_hard_eq1} must be achieved at binary $y$, i.e., $y \in \{0,1\}^n$.
\end{description}

We construct an instance of~\eqref{eqn:dual_oracle} that is equivalent to an instance of \eqref{eqn:single_hard_eq1}. We first scale the variables so that the feasible region is a subset of $\Delta_n$. Define $x = \frac{1}{n} y$, we have the following NP-hard problem,
\begin{align} \label{eqn:quadratic}
\max \quad & (\beta^T x)^2 - \gamma^T x\\
\text{s.t.} \quad & Ax \le b \notag\\
& 0 \le x_i \le \frac{1}{n}, \quad \forall i\in[n] \notag
\end{align}
where $\gamma, \beta$ are non-negative vectors with length $n$. We add a variable $z = 1 - \sum_i x_i$, so that \eqref{eqn:quadratic} is equivalent to 
\begin{align*}
\max \quad & (\beta^T x)^2 - \gamma^T x \notag\\
\text{s.t.} \quad & Ax \le b \notag\\
& 0 \le x_i \le \frac{1}{n}, \quad \forall i\in[n] \notag\\
&\bm{1}^T x + z = 1. \notag
\end{align*}
Then we replace all the constants by their products with $\bm{1}^T x + z$, which is equal to $1$,
\begin{align} \label{eqn:quadratic2}
\max_{(x,z) \in \Delta_{n+1}} & \quad (\beta^T x)^2 - \gamma^T x\\
\text{s.t.} & \quad Ax - b \left(\bm{1}^T x+ z \right)\le 0 \notag\\
& \quad x_i - \frac{1}{n} \left(\bm{1}^T x + z \right) \le 0, \quad \forall i\in[n]. \notag
\end{align}
\eqref{eqn:quadratic2} is an instance of \eqref{eqn:dual_oracle} by letting $y=(x,z)$, $R(y) = C_a(y) = \beta^T x + 0 \cdot z$, $\alpha = (\gamma, 0)$, and defining a bunch of $C_{a'}(y)$ so that linear constraints $C_a(y) - C_{a'}(y) \le 0$ equal the linear constraints in~\eqref{eqn:quadratic2}. 
\end{proof}

We then show that it is hard to solve 
\begin{align} \label{eqn:dual_prob}
    \max_{y\in \Delta_n} \{ R(y)\cdot C(y) -  \alpha^T y \}.
\end{align}
\begin{lemma}
There exist fixed $R(y) = r^T y$, fixed $\alpha$ and partially fixed $C(y)$, such that deciding whether the solution of~\eqref{eqn:dual_prob} is greater or equal to $0$ is NP-complete, and the maximum must be achieved at point $y$ with
$$
y_j \in \{0, 1/n\} \text{ for } 1\le j \le n-1, \quad y_n = 1 - \sum_{j=1}^{n-1} y_j.
$$
\end{lemma}
\begin{proof}
	We construct an instance of \eqref{eqn:dual_prob} that is equivalent to \eqref{eqn:dual_oracle}. Suppose we have an instance of \eqref{eqn:dual_oracle}
	\begin{align*} 
    \max_{y\in \Delta_n} \{ R(y)\cdot C_a(y) -  \alpha^T y ~|~ C_a(y) \le C_{a'}(y), \forall a'\}.
\end{align*}

We show that we can change $C_{a'}(y)$ for $a'\neq a$ so that the maximum of function $R(y)\cdot C(y) -  \alpha^T y$ cannot lie within region 
$$
C_{a'}(y) < C_a(y), \text{  i.e.,  } (c_{a'}- c_a)^T y < 0
$$
for any $a'$.
The gradient of $R(y)\cdot C(y) -  \alpha^T y$ is equal to
\begin{align*}
\nabla 	(R(y)\cdot C(y) -  \alpha^T y ) = C(y)\cdot r + R(y) \cdot \nabla C(y)  - \alpha.
\end{align*}
Consider a point $y_0$ with $C(y_0) = C_{a'}(y_0)$, so $\nabla C(y_0) = c_{a'}$ and $(c_{a'}- c_a)^T y_0 < 0$. Consider the projection of $y_0$ to the line $(c_{a'}- c_a)^T y = 0$. Let the projection be $y^*$. Then we should have $y^* - y_0 \propto c_{a'}- c_a$.
The directional derivative of $R(y)\cdot C(y) -  \alpha^T y$ along direction $d \propto c_{a'} -c_{a}$ at this point is equal to 
\begin{align}
d^T(C(y_0)\cdot r + R(y_0) \cdot c_{a'}  - \alpha) & = d^T(C_{a'}(y_0)\cdot r + R(y_0) \cdot c_{a} + R(y_0)(c_{a'} - c_a) - \alpha)\notag\\
& = (C_{a'}(y_0)d^T r + R(y_0)d^T c_a - d^T \alpha ) - R(y_0) \Vert c_{a'} - c_a\Vert \notag\\
& \le (C_{a}(y_0)d^T r + R(y_0)d^T c_a - d^T \alpha ) - R(y_0) \Vert c_{a'} - c_a\Vert. \label{eqn:single_hard_dd}
\end{align}
Note that we can increase  $\Vert c_{a'} - c_a\Vert$ by a factor of $k$ without changing $c_a$ and region $\{y: C_{a'}(y) < C_a(y)\}$ (or equivalently the direction of $c_{a'}- c_a$) by replacing $c_{a'}$ with
$$
\widetilde{c}_{a'} = k (c_{a'}- c_a) + c_a
$$
so that $\widetilde{c}_{a'}- c_a = k (c_{a'}- c_a)$. Since $R(y_0) = r^T y_0 \ge \min_i \{r_i\}$ is bounded from below by a positive constant, we can choose $k$ that is large enough so that  the directional derivative~\eqref{eqn:single_hard_dd} is negative for all the points between $y_0$ and $y^*$, which means that $R(y_0)C(y_0) - \alpha^T y_0 < R(y^*)C(y^*) - \alpha^T y^*$. So $y_0$ cannot be the maximum point. 

So we replace $c_{a'}$ with $\widetilde{c}_{a'}$  for all $a' \neq a$ to get the new cost of uncertainty function $\widetilde{C}(y)$  so that the maximum 
\begin{align*} 
    \max_{y\in \Delta_n} \{ R(y)\cdot \widetilde{C}(y) -  \alpha^T y \}
\end{align*}
cannot lie in region
$$
C_{a'}(y) < C_a(y), \text{  i.e.,  } (c_{a'}- c_a)^T y < 0
$$
for any $a'$, which is thus equal to 
	\begin{align*} 
    \max_{y\in \Delta_n} \{ R(y)\cdot C_a(y) -  \alpha^T y ~|~ C_a(y) \le C_{a'}(y), \forall a'\}
\end{align*}

\end{proof}

Finally we prove that Problem \eqref{eqn:dual_prob}
\begin{align*} 
    \max_{y\in \Delta_n} \{ R(y)\cdot C(y) -  \alpha^T y \} = \max_{y\in \Delta_n} \{ f(y) -  \alpha^T y \}
\end{align*}
whose maximum must be achieved at point $y$ with
\begin{align} \label{eqn:single_binary_p}
y_j \in \{0, 1/n\} \text{ for } 1\le j \le n-1, \quad y_n = 1 - \sum_{j=1}^{n-1} y_j
\end{align}
 is polynomial-time reducible to the seller's problem of finding concave closure $\overline{f}$ at a point \eqref{eqn:CC_dual}.
First by the definition of concave closure, it holds that  
$$
\max_{y\in \Delta_n} \{ f(y) -  \alpha^T y \} = \max_{y\in \Delta_n} \{ \overline{f}(y) -  \alpha^T y \}.
$$
So to solve \eqref{eqn:dual_prob}, we only need to find the maximum of concave function $\overline{f}(y) -  \alpha^T y$ at points that satisfy~\eqref{eqn:single_binary_p}. If we have an oracle that solves the concave closure  $\overline{f}(y)$ defined in \eqref{eqn:CC_dual}, then we can use the ellipsoid method to find the maximum of $\overline{f}(y) -  \alpha^T y$. When we know the maximum point must have~\eqref{eqn:single_binary_p} and only need to find the maximum at points with~\eqref{eqn:single_binary_p}, the ellipsoid method can terminate within polynomially many iterations, by the same arguments as in~\cite{goemans2005lecture}. Therefore the NP-hard problem \eqref{eqn:dual_prob} is polynomial-time reducible to finding the concave closure $\overline{f}$. Therefore the seller's optimal information advertising problem is NP-hard.



\subsection{Proof of Lemma~\ref{lem:opt_size}} \label{app:single_opt_size}

Let $\langle S, \pi \rangle$ be an optimal advertising rule. Suppose  $S = \{s_1, \dots, s_k\}$ with $k > n$. For simplicity, we write $\eta^{(i)} = \eta^{s_i}$ as the posterior when $s_i$ is received, and $\phi^{(i)} = \phi_\theta(s_i)$ as the probability of receiving $s$ based on the buyer's belief $\theta$. WLOG assume $\phi^{(i)} > 0$ for all $1\le i \le k$. Since $k>n$ we must have $\eta^{(1)}, \dots, \eta^{(k)}$ linearly dependent. So there exists non-zero vector $\alpha$ with 
$$
\alpha_1 \eta^{(1)}+ \cdots + \alpha_k \eta^{(k)} = 0.
$$
WLOG assume $\alpha_1 \neq 0$.
Then 
$$
\eta^{(1)} = - \frac{\alpha_2}{\alpha_1} \eta^{(2)} - \cdots - \frac{\alpha_k}{\alpha_1} \eta^{(k)}. 
$$
We can then try to reduce the size of $S$ by substituting $\eta^{(1)}$ with $- \frac{\alpha_2}{\alpha_1} \eta^{(2)} - \cdots - \frac{\alpha_k}{\alpha_1} \eta^{(k)}$, that is, reducing $\phi^{(1)}$ by $\delta$ and increasing other  $\phi^{(i)}$ by $- \frac{\alpha_i}{\alpha_1}\delta$. This will not violate the constraints of~\eqref{prog:opt2}. We increase the value of $\delta$ until one of $\phi^{(i)}$ reaches $0$ and that signal can be removed from $S$.
Since we are considering an optimal advertising rule, we must have the value of $f(x) = R(x) C(x)$ satisfying
$$
f(\eta^{(1)}) = - \frac{\alpha_2}{\alpha_1} f(\eta^{(2)}) - \cdots - \frac{\alpha_k}{\alpha_1} f(\eta^{(k)}),
$$
otherwise we can substitute one of $\eta^{(1)}$ and $- \frac{\alpha_2}{\alpha_1} \eta^{(2)} - \cdots - \frac{\alpha_k}{\alpha_1} \eta^{(k)}$ with another to strictly increase the objective value without violating the constraints.
Therefore as long as $k>n$, we can reduce the size of $S$ by one without violating the constraints or changing the objective value. 

\subsection{Proof of Lemma~\ref{lem:opt_decompose} }\label{app:single_opt_decompose}
Suppose there  exist $\eta^{(1)},\eta^{(2)} \in \mathcal{P}_a$ with 
	 $
	 \eta^s = \alpha \eta^{(1)} + (1-\alpha) \eta^{(2)}
	 $ 
	 and 
	 $
	 (R(\eta^s) - R(\eta^{(1)})) (C(\eta^s) - C(\eta^{(1)}))>0.
	 $
Since both $R(\eta)$ and $C(\eta)$ are linear functions within $\mathcal{P}_a$, we should have 
$$
(R(\eta^s), C(\eta^s)) = \alpha (R(\eta^{(1)}), C(\eta^{(1)})) + (1-\alpha) (R(\eta^{(2)}), C(\eta^{(2)})),
$$
and $(R(\eta^s), C(\eta^s))$, $(R(\eta^{(1)}), C(\eta^{(1)}))$, $(R(\eta^{(2)}), C(\eta^{(2)}))$ lying on a line with a positive slope. Since function $g(x,y) = xy$ is strictly convex along a direction with positive slope, decomposing $\eta^s$ into $\alpha \eta^{(1)} + (1-\alpha) \eta^{(2)}$ should strictly increase the objective value of~\eqref{prog:opt2}, which contradicts the optimality of the advertising rule.

\subsection{Proof of Lemma~\ref{lem:opt_never}} \label{app:lem_never}

The idea is as follows.
	Consider merging $s$ and $t$ into one single $r$ so that 
	$$\eta^{r} = \alpha \eta^{s} + (1-\alpha) \eta^{t}, \quad \alpha\in(0,1) $$
	as in~\eqref{eqn:single_merge}. First assuming that $C(\eta)$ is a linear function of $\eta$, then we should have $\big(R(\eta^s), C(\eta^s)\big)$, $\big(R(\eta^{r}), C(\eta^{r})\big)$, $\big(R(\eta^{t}), C(\eta^{t})\big)$ lying on a line with a negative slope. Since $g(x,y) = xy$ is a strictly concave function along a direction with a negative slope, merging $s$ and $t$ should lead to a higher objective value assuming $C(\cdot)$ is linear,
	$$
	R(\eta^r) (\alpha  C(\eta^{s}) + (1-\alpha) C(\eta^{t})) > \alpha R(\eta^{s}) C(\eta^{s}) + (1-\alpha) R(\eta^{t})C(\eta^{t}).
	$$
	When $C(\cdot)$ is concave but not linear, the gap will only be larger as $C(\eta^r)\ge  
	\alpha  C(\eta^{s}) + (1-\alpha) C(\eta^{t})$. So merging $s$ and $t$ will still increase the objective value,
	$$
	R(\eta^r) C(\eta^r)\ge  
	R(\eta^r) (\alpha  C(\eta^{s}) + (1-\alpha) C(\eta^{t})) > \alpha R(\eta^{s}) C(\eta^{s}) + (1-\alpha) R(\eta^{t})C(\eta^{t}),
	$$
	which contradicts the optimality of the advertising rule.

One can compute the follows. Consider  merging two signals, that is, having a new signal $v$ with 
\begin{align*}
\pi(\omega, v) = \pi(\omega, s) + \pi(\omega, t), \quad \forall \omega.
\end{align*}
It is easy to verify that 
\begin{align*}
& \phi_\theta(v) = \phi_\theta(s) + \phi_\theta(t),\\
& \eta^v = \frac{\phi_\theta(s)}{\phi_\theta(v)} \cdot \eta^s + \frac{\phi_\theta(t)}{\phi_\theta(v)} \cdot \eta^t,\\
& R(\eta^v) = \frac{\phi_\theta(s)}{\phi_\theta(v)} \cdot R(\eta^s) + \frac{\phi_\theta(t)}{\phi_\theta(v)} \cdot R(\eta^t).
\end{align*}
 Then the change of the expected revenue when merging $s$ and $t$ is equal to
\begin{align}
&\phi_\theta(v) R(\eta^v) C(\eta^v) - \phi_\theta(s) R(\eta^s) C(\eta^s) - \phi_\theta(t) R(\eta^t) C(\eta^t) \notag \\
= & \big(\phi_\theta(s) \cdot R(\eta^s) + \phi_\theta(t) \cdot R(\eta^t) \big) C(\eta^v) - \phi_\theta(s) R(\eta^s) C(\eta^s) - \phi_\theta(t) R(\eta^t) C(\eta^t). \label{eqn:chrc_diff_eq1}
\end{align}\label{eqn:chrc_diff_eq2}
Since $C(\cdot)$ is a concave function, we have 
\begin{align}
C(\eta^v) = C \left( \frac{\phi_\theta(s)}{\phi_\theta(v)} \cdot \eta^s + \frac{\phi_\theta(t)}{\phi_\theta(v)} \cdot \eta^t \right) \ge \frac{\phi_\theta(s)}{\phi_\theta(v)} \cdot C(\eta^s) + \frac{\phi_\theta(t)}{\phi_\theta(v)} \cdot C(\eta^t).
\end{align}
Combining \eqref{eqn:chrc_diff_eq1} and \eqref{eqn:chrc_diff_eq2}, we know that the change of the expected revenue is no less than
\begin{align*}
&\big(\phi_\theta(s) \cdot R(\eta^s) + \phi_\theta(t) \cdot R(\eta^t) \big) \left( \frac{\phi_\theta(s)}{\phi_\theta(v)} \cdot C(\eta^s) + \frac{\phi_\theta(t)}{\phi_\theta(v)} \cdot C(\eta^t)\right) - \phi_\theta(s) R(\eta^s) C(\eta^s) - \phi_\theta(t) R(\eta^t) C(\eta^t)\\
= & \frac{ \left(\phi_\theta(s)  R(\eta^s) + \phi_\theta(t) R(\eta^t) \right) \left( \phi_\theta(s)  C(\eta^s) + \phi_\theta(t) C(\eta^t)\right) - \phi_\theta(s) \phi_\theta(v) R(\eta^s) C(\eta^s) - \phi_\theta(t) \phi_\theta(v) R(\eta^t) C(\eta^t) }{\phi_\theta(v)}\\
= & \frac{\phi_\theta(s) \phi_\theta(t)}{\phi_\theta(v)} \cdot \left( R(\eta^s) C(\eta^t) +  R(\eta^t)  C(\eta^s) - R(\eta^s) C(\eta^s) - R(\eta^t) C(\eta^t) \right) \\
= & - \frac{\phi_\theta(s) \phi_\theta(t)}{\phi_\theta(v)} \cdot \big( R(\eta^s) - R(\eta^t)\big) \big(C(\eta^s) - C(\eta^t)\big)\\
> & 0.
\end{align*}

\subsection{Proof of Lemma~\ref{lem:point_pair}} \label{app:point_pair}

We first restate some notations.
$$
\mathcal{P}_a = \{ \eta: C_a(\eta) \le C_{a'}(\eta), \forall a'\}
$$
is the polytope in which action $a$ is always the best action and $C(\theta) = C_a(\theta)$ is linear. $\mathcal{H}_a$ is the set of vertices of the polytope $\mathcal{P}_a$,
Define $\mathcal{Q}_a = \{(R(\eta), C(\eta)): \eta \in \mathcal{P}_a\}$.

We prove that there exists an optimal advertising rule with each $\eta^s$ lying on the segments between the vertices in $\mathcal{H}_a$ for some $a$.

Consider an optimal advertising rule $\langle S, \pi\rangle$ and a signal $s\in S$ with $\phi_\theta(s)>0$. Suppose $\eta^s$ lies in $\mathcal{P}_a$, then $\eta^s$ can be represented as a convex combination of the vertices of the polytope, 
$$
\eta^s = \sum_{i \in \mathcal{H}_a} q_i \cdot i.
$$
Let $T\subseteq \mathcal{H}_a$ be the set of vertices $i$ that has $q_i > 0$. If $|T| \le 2$, then the lemma is proved. Otherwise we claim the follows. 
\begin{claim} \label{clm:nonp_line_app}
The points $\{(R(i), C(i)): i \in T\}$ in two-dimensional space, which represent the likelihood ratio and the cost of uncertainty of $i \in T$, must lie on a line with a nonpositive slope. 
\end{claim}

If the points $\{(R(i), C(i)): i \in T\}$ does not lie on a line with a nonpositive slope, there are two possibilities,
\begin{enumerate}
\item the points lie on a line with a positive slope,
\item the points do not lie on a line.	
\end{enumerate}
In both of the cases, we can decompose $(R(\eta^s), C(\eta^s))$ along a direction $d$ with positive slope (as shown in Figure~\ref{fig:point_pair_app}).  
\begin{itemize}
\item In Case (1), when $\{(R(i), C(i)): i \in T\}$ lie on a line with a positive slope, suppose $(R(l), C(l))$ with $l \in T$ is one of the endpoints of segment $\mathcal{Q}_a$. Since $\eta^s = \sum_{i \in \mathcal{H}_a} q_i \cdot i$ and $q_i \in (0,1)$ for all i,  there exists small enough $\varepsilon$ so that $$\eta^s + \varepsilon (l-\eta^s) \in \mathcal{P}_a$$
 $$\eta^s - \varepsilon (l-\eta^s) = (q_l + (q_l-1) \varepsilon) l + \sum_{i \neq l} (1+\varepsilon)q_i \cdot i  \in \mathcal{P}_a $$
Therefore $\eta^s$ can be decomposed as 
$$
\eta^s = \frac{1}{2} (\eta^s + \varepsilon (l - \eta^s)) + \frac{1}{2} (\eta^s - \varepsilon (l - \eta^s)).
$$
\item In Case (2), there must exist a convex combination $l = \sum_{i\in T} w_i \cdot i$ so that $(R(l) - R(\eta^s))(C(l) -C(\eta^s)) > 0$. Again there must exist small enough $\varepsilon$ such that 
 $$
 \eta^s + \varepsilon (l-\eta^s) \in \mathcal{P}_a
 $$
 $$
 \eta^s - \varepsilon (l-\eta^s) = \sum_{i\in T} (q_i + \varepsilon( q_i - w_i)) \cdot i \in \mathcal{P}_a
 $$
 Then $\eta^s$ can be decomposed as 
$$
\eta^s = \frac{1}{2} (\eta^s + \varepsilon (l - \eta^s)) + \frac{1}{2} (\eta^s - \varepsilon (l - \eta^s)).
$$
\end{itemize}

Therefore according to Lemma~\ref{lem:opt_decompose}, both of the cases cannot be true for an optimal mechanism.

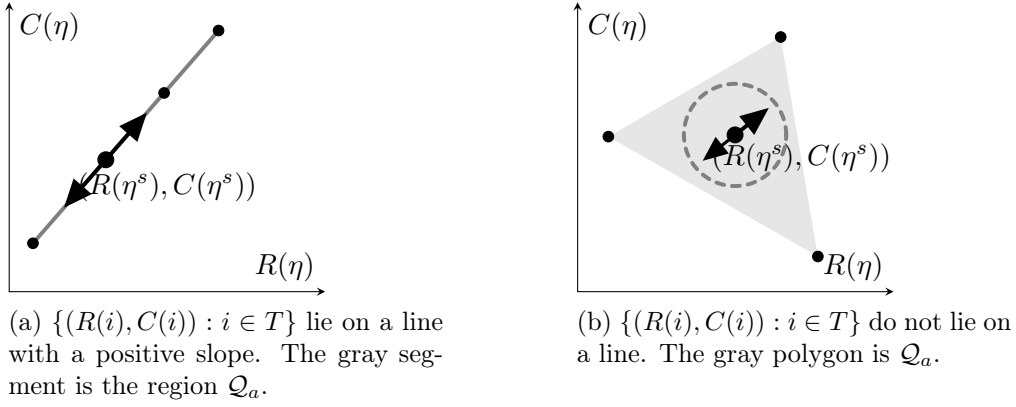
\begin{figure}
\centering
\begin{subfigure}[t]{0.35\textwidth}
\begin{tikzpicture}[line cap=round,line join=round,>=triangle 45,x=0.35cm,y=0.35cm]
\begin{axis}[
x=0.35cm,y=0.35cm,
axis lines=middle,
xmin=2,
xmax=14,
ymin=0,
ymax=11,
xtick={0},
ytick={0},
xlabel={$ R(\eta)$},
ylabel={$C(\eta)$}
]
\clip(-1.2390112325856555,-0.4135013389416602) rectangle (14.99975892730601,12.905939354003209);
\draw [line width=1.5pt, color=gray] (2.902339879923165,1.8397004666703274) -- (9.951751724043808,9.925790523161664);
\draw [->,line width=1.5pt] (5.66926878915395,5.013530686082114) -- (7.2296251254594175,6.8033511894913286);
\draw [->,line width=1.5pt] (5.66926878915395,5.013530686082114) -- (4.082649029366113,3.1935844910313573);
\begin{scriptsize}
\draw [fill=black] (9.951751724043808,9.925790523161664) circle (2pt);
\draw [fill=black] (7.883893351623145,7.553835331267372) circle (2pt);
\draw [fill=black] (2.902339879923165,1.8397004666703274) circle (2pt);
\draw [fill=black] (5.66926878915395,5.013530686082114) circle (3pt);
\draw[color=black] (8,4) node {{$(R(\eta^s), C(\eta^s))$}};
\end{scriptsize}
\end{axis}
\end{tikzpicture}
\caption{$\{(R(i), C(i)): i \in T\}$ lie on a line with a positive slope. The gray segment is the region $\mathcal{Q}_a$.}
\end{subfigure}\hspace{0.1 \textwidth}
\begin{subfigure}[t]{0.35\textwidth}
\begin{tikzpicture}[line cap=round,line join=round,>=triangle 45,x=0.35cm,y=0.35cm]
\begin{axis}[
x=0.35cm,y=0.35cm,
axis lines=middle,
xmin=1,
xmax=13,
ymin=0,
ymax=11,
xtick={0},
ytick={0},
xlabel={$ R(\eta)$},
ylabel={$C(\eta)$}
]
\clip(-1.2390112325856555,-0.4743207028363856) rectangle (14.99975892730601,12.905939354003209);
\fill[line width=2pt,fill=black,fill opacity=0.10000000149011612] (2.167732242368716,5.894053032749113) -- (8.713461588084149,9.680308242525713) -- (10.125285564611007,1.3377120175942223) -- cycle;
\draw [gray, line width=1.5pt,dashed] (6.980768525983005,5.958226849863971) circle (0.68cm);
\draw [->,line width=1.5pt] (6.980768525983005,5.958226849863971) -- (8.26424486828015,6.985007923701692);
\draw [->,line width=1.2pt] (6.980768525983005,5.958226849863971) -- (5.748192729107704,4.972166212363726);
\begin{scriptsize}
\draw [fill=black] (2.167732242368716,5.894053032749113) circle (2pt);
\draw [fill=black] (8.713461588084149,9.680308242525713) circle (2pt);
\draw [fill=black] (10.125285564611007,1.3377120175942223) circle (2pt);
\draw [fill=black] (6.980768525983005,5.958226849863971) circle (3pt);
\draw[color=black] (9.5,5.1) node {$(R(\eta^s), C(\eta^s))$};
\end{scriptsize}
\end{axis}
\end{tikzpicture}
\caption{$\{(R(i), C(i)): i \in T\}$ do not lie on a line. The gray polygon is  $\mathcal{Q}_a$.}
\end{subfigure}

\caption{An illustration for Claim~\ref{clm:nonp_line_app}. Two cases when $\{(R(i), C(i)): i \in T\}$ do not lie on a line with a nonpositive slope. The big black point represents $(R(\eta^s), C(\eta^s))$, and the small black points are $(R(i), C(i))$ for $i\in T$. $(R(\eta^s), C(\eta^s))$ is a convex combination of the points $\{(R(i), C(i)): i \in T\}$, with positive coefficients. In both of the cases, $(R(\eta^s), C(\eta^s))$ can be decomposed along a direction with a positive slope.}
\label{fig:point_pair_app}
\end{figure}

Based on Claim~\ref{clm:nonp_line_app}, we know that $\{(R(i), C(i)): i \in T\}$ must lie on a line with a nonpositive slope. Then we show that if $|T|>2$, we can decompose signal $s$ to a bunch of signals that have $\eta$ lying on segments between vertices. More specifically, we can decompose $\eta^s$ into a convex combination of some points with the same likelihood ratio and cost of uncertainty, 
$$
\eta^s = \alpha_1 \eta^{(1)} + \alpha_2 \eta^{(2)} + \cdots + \alpha_k \eta^{(k)}.
$$
with each point $\eta^{(l)} = \beta \cdot i + (1-\beta) \cdot j$ for some $i,j \in T$ and $(R(\eta^{(l)}), C(\eta^{(l)})) = (R(\eta^s), C(\eta^s))$, for $1\le l \le k$. We find $\eta^{(1)}, \eta^{(2)},  \dots,  \eta^{(k)}$ by repeating the following process
\begin{itemize}
	\item At step $l$, let $i^*,j^* \in T$ be the two endpoints on the segment of $\{(R(i), C(i)): i \in T\}$. Then there exists $\beta^* \in (0,1)$ such that $ \beta^* (R(i^*), C(i^*)) + (1- \beta^*) (R(j^*), C(j^*)) = (R(\eta^s), C(\eta^s))$. 
	\item Let $\eta^{(l)} = \beta^* \cdot i^* + (1-\beta^*) \cdot j^*$, $\alpha_l = \min\{ q_{i^*}/ \beta^*, q_{j^*}/(1-\beta^*)\}$. WLOG assume $\alpha_l = q_{i^*}/ \beta^*$.
	\item $\eta^s \gets \eta^s -  \alpha_l \eta^{(l)}$, $T \gets T \setminus \{i^*\}$.
\end{itemize}
At each step, we find an $\eta^{(l)}$ and reduce the size of $T$ at least by one. Repeat this process until $|T| \le 2$, we find an advertising rule that has each posterior lying on the segments between $i,j \in T \subseteq \mathcal{H}_a$. 
All the points $i \in T$ have $(R(i), C(i))$ lying on a line with a nonpositive slope, so we have $(R(i^*) - R(j^*))(C(i^*) - C(j^*)) \le 0$ for all $i^*$, $j^*$.

It remains to prove that for each pair $i,j$, there is a unique $\eta^s$ lying on the segment between $i,j$. Suppose there are two posteriors $\eta^s, \eta^t$ lying on the segment between $i,j$ with $(R(i) - R(j))(C(i) - C(j)) \le 0$. If $(R(i) - R(j))(C(i) - C(j)) < 0$, then by Lemma~\ref{lem:opt_never}, we must have $(R(\eta^s), C(\eta^s)) = (R(\eta^t), C(\eta^t))$. So we can merge $s$ and $t$ into one signal without changing the objective value. If $(R(i) - R(j))(C(i) - C(j)) = 0$, then either $R(\eta^s) = R(\eta^t)$ or $C(\eta^s) = C(\eta^t)$. In both of the cases, merging $s$ and $t$ will not decrease the objective value as $R(\cdot)$ is a linear function and $C(\cdot)$ is a concave function.

\subsection{Finding the extreme points} \label{app:extreme}

To find the vertices of the polytope $\mathcal{P}_a$, we start with the linear constraints that specify  $\mathcal{P}_a$. Recall that $\mathcal{P}_a \subseteq \Delta \Omega$ is the set of posterior beliefs based on which action $a$ is the best action. $\mathcal{P}_a$ can be defined by linear equations with non-negative variables as
	\begin{align}
	&  \Vert \eta \Vert_1 = 1 \notag\\
	& C_a(\eta) - C_{a'}(\eta) + s_{a'} =  0, \quad \forall a' \neq a \label{prog:basic_solution}\\
	& \eta\ge 0, \quad s_{a'} \ge 0 \notag
	\end{align}
Variables $s_{a'}$ are the slack variables that are added to convert inequality constraints $C_a(\eta) - C_{a'}(\eta) \le 0$ into equality constraints. Then the vertices of $\mathcal{P}_a$ can be found by solving the \emph{basic feasible solutions} of \eqref{prog:basic_solution}. 
\begin{definition} \label{def:basic_fs}
Let $\mathcal{P}$ be a polytope defined by $\mathcal{P} = \{ x: A x = b, x\ge 0\}$, where $A$ is a $m\times n$ matrix with $m\le n$. Without loss of generality we  assume rank$(A) = m$.\footnote{If rank$(A) < m$, then there exist redundant constraints that can be identified and removed.} Then a basic feasible solution is a solution $\in \mathcal{P}$ with $n-m$ variables set to zero. These $n-m$ zero variables are called non-basic variables of the solution, and the other $m$ variables are called basic variables.  
\end{definition}
\begin{lemma}[Theorem 2.3 in~\cite{bertsimas1997introduction}] \label{lem:basic_fs}
    Let $\mathcal{P}$ be a polytope defined by $\mathcal{P} = \{ x: A x = b, x\ge 0\}$. Then $x$ is a vertex of $\mathcal{P}$ if and only if $x$ is a basic feasible solution of $\{ x: A x = b, x\ge 0\}$.
\end{lemma}
We can find all the basic feasible solutions of $\mathcal{P} = \{ x: A x = b, x\ge 0\}$ as follows. Let $B \subseteq [n]$ be a set of indices that correspond to $m$ linearly independent columns of the matrix $A$. We can then represent matrix $A$ as the concatenation of two matrices $A=[A_B\,|\,A_N]$ where $A_B$ is the $m \times m$ matrix whose columns are indexed by the indices in $B$,  and $A_N$ is the $m \times (n-m)$ matrix whose columns are indexed by the indices in $[n] \setminus B$. Then we have the following lemma.
\begin{lemma}[Theorem 2.3 in~\cite{bertsimas1997introduction}]
For any basic feasible solution $x$, we have a set $B \subseteq [n]$ of $m$ indices that correspond to a linearly independent set of columns of $A$ such that:	(1) basic variables $x_B = A_B^{-1} b$; (2) non-basic variables $x_N = 0$ where $N = [n] \setminus B$. In addition, for any set $B \subseteq [n]$ of $m$ indices that correspond to a linearly independent set of columns, if $x_B = A^{-1}_B b \ge 0$ then $(x_B, x_N=0)$ is a basic feasible solution.
\end{lemma}
Therefore we can find all the basic feasible solutions by enumerating $B\subseteq [n]$ and computing $x_B = A^{-1}_B b$ and checking whether $x_B\ge 0$. 
In general, there are exponentially many possible $B\subseteq [n]$ and thus finding all the vertices would take exponential time. However when $|A|$ is a constant or $|\Omega|$ is a constant, the number of possible $B$ would not be very large.  
\begin{lemma} \label{lem:num_vertices}
	When $|A|$ is a constant or $|\Omega|$ is a constant, $|\mathcal{H}_a| = poly(|\Omega|, |A|)$ for all $a$. 
\end{lemma}
\begin{proof}
	$\mathcal{H}_a$, the set of extreme points of the polytope $\mathcal{P}_a$, is the set of the basic feasible solutions of \eqref{prog:basic_solution},
which contains $|A|$ constraints and $|\Omega| + |A| -1$ variables. So there are at most $C(|\Omega| + |A| -1, |A|)$ basic feasible solutions (which is the number of possible choices of basic variables $B$). When $|\Omega|$ or $|A|$ is a constant, $C(|\Omega| + |A| -1, |A|) = poly(|\Omega|, |A|)$.
\end{proof}
\begin{theorem}
	There exists an optimal advertising rule $\langle S, \pi \rangle$ with $|S| \le n$, in which the seller reveals  $\le 2|A|$ possibilities of the realized state of the world $\omega$ to the buyer before selling the (remaining) information. Formally, for all $s \in S$, the buyer's posterior $\eta^s$ has no more than $2 |A|$ non-zero entries,
	$$
	\Vert \eta^s \Vert_0 \le 2|A|.
	$$
\end{theorem}
\begin{proof}
Recall that the set of vertices $\mathcal{H}_a$ is a subset of the basic feasible solutions of~\eqref{prog:basic_solution}. And each basic feasible solution has $|A|$ non-zero variables, which means that the points in $\mathcal{H}_a$ have no more than $|A|$ non-zero entries. According to Lemma~\ref{lem:point_pair}, there exists an optimal advertising rule with $\eta^s = \beta \cdot i + (1-\beta) j$ for $i,j\in \mathcal{H}_a, a\in A$, which means that $\eta^s$ will not have more than $2|A|$ non-zero entries. 
\end{proof}

\begin{proposition} 
When $\theta = \mu$, there exists an optimal advertising rule $\langle S, \pi \rangle$ with $|S| \le n$, that reveals $\le |A|$ possibilities of the realized state of the world $\omega$ to the buyer before selling the (remaining) information, i.e., for all $s\in S$,
$$
	\Vert \eta^s \Vert_0 \le |A|.
	$$
\end{proposition}
\begin{proof}
We prove that when $\theta = \mu$, there exists an optimal advertising rule that reveals $\le |A|$ possibilities of the realized state of the world $\omega$ to the buyer before selling the (remaining) information. 
Recall that in this case, the optimal advertising problem is equivalent to finding the concave closure of $f(x) = C(x)$,
\begin{align*} 
\max_{\eta,\,\phi_\theta} &\quad \sum_{s\in S} \phi_\theta(s) \cdot C(\eta^s)\\
\text{s.t.} & \quad \sum_{s\in S} \phi_\theta(s) \cdot \eta^s = \theta \notag \\
& \quad \phi_\theta(s) \ge 0, \ \eta^s \in \Delta \Omega, \ \forall s \notag
\end{align*}
The optimal objective value is $C(\theta)$ being achieved at
$$
S = \{s\}, \ \eta^s = \theta, \ \phi_\theta(s) = 1.
$$ 
which means is optimal to not give any advertising information and directly charge the buyer his expected gain. We show that $\eta^s = \theta$ can be decomposed into points in $\mathcal{H}_a$ for some $a\in A$ without changing the expected revenue. Assume that $\theta \in \mathcal{P}_a$, then $\theta$ can be decomposed into a convex combination of the vertices in $\mathcal{H}_a$,
$$
\theta = \sum_{i\in \mathcal{H}_a} q_i \cdot i.
$$
Since $C(x)$ is linear within $\mathcal{P}_a$,  we can decompose $\theta$ into $\sum_{i\in \mathcal{H}_a}$ as
$$
S= \{s_i: i \in \mathcal{H}_a\}, \ \eta^{s_i} = i, \ \phi_\theta(s_i) = q_i,
$$
and the objective function value remains unchanged
$$
\sum_{s_i \in S} \phi_\theta(s) C(\eta^s) = \sum_{i \in \mathcal{H}_a} q_i C(i) =  C\left(\sum_{i \in \mathcal{H}_a} q_i  \right) = C(\theta)
$$
The vertices in $\mathcal{H}_a$ has no more than $|A|$ non-zero entries by Definition~\ref{def:basic_fs} and Lemma~\ref{lem:basic_fs}.
Finally we can decrease the number of signals $|S|$ to $\le n$ by the same method in the proof of Lemma~\ref{lem:opt_size}.

\end{proof}

\subsection{Proof of Theorem~\ref{thm:single_binary}} \label{app:single_binary}

We prove that when $|\Omega| = 2$, 
there exists an optimal advertising rule $\langle S, \pi\rangle$ with $S= \{ s, t\}$. The optimal advertising rule has $\eta^s_1$ being a vertex of $C(\eta_1)$ on one side of $\eta^*_1$, and $\eta^t_1$ lying on the other side of $\eta^*_1$, as illustrated in Figure~\ref{fig:single_binary_app}. The optimal advertising rule can be solved in $O(|A|^2)$ time.
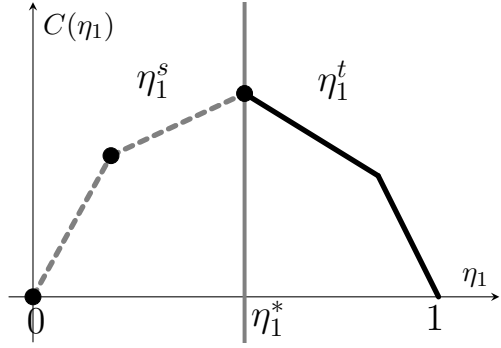
\begin{figure}
\centering
\begin{tikzpicture}[line cap=round,line join=round,>=triangle 45,x=0.4cm,y=0.4cm]
\begin{axis}[
x=0.4cm,y=0.4cm,
axis lines=middle,
xmin=-0.8,
xmax=15.5,
ymin=-1.5,
ymax=9.8,
xtick={0},
ytick={0},
xlabel={$\eta_1$},
ylabel={$C(\eta_1)$}]
\clip(-1.2390112325856555,-3.028733986414866) rectangle (14.99975892730601,10.108248614845827);
\draw [line width=2pt, dashed, color=gray] (0,0)-- (2.592608692782041,4.695325228215276);
\draw [line width=2pt, dashed, color=gray] (2.592608692782041,4.695325228215276)-- (7.032422257096991,6.763183600635941);
\draw [line width=2pt] (7.032422257096991,6.763183600635941)-- (11.47223582141194,4.026312225373297);
\draw [line width=2pt] (11.47223582141194,4.026312225373297)-- (13.479274829937877,0);
\draw [line width=1.5pt, color=gray] (7.032422257096991,-3.028733986414866) -- (7.032422257096991,10.108248614845827);
\begin{scriptsize}
\draw [fill=black] (0,0) circle (3pt);
\draw[color=black] (0.1,-0.8) node {\Large $0$};
\draw [fill=black] (2.592608692782041,4.695325228215276) circle (3pt);
\draw [fill=black] (7.032422257096991,6.763183600635941) circle (3pt);
\draw[color=black] (13.357636102148426,-0.7) node {\Large $1$};
\draw[color=black] (7.82307398772842,-0.7) node {\Large $\eta^*_1$};
\draw[color=black] (4,7.2) node {\Large $\eta^s_1$};
\draw[color=black] (10,7.2) node {\Large $\eta^t_1$};
\end{scriptsize}
\end{axis}
\end{tikzpicture}

\caption{Characterization of the optimal advertising rule when the state of the world is binary. We plot the function $C(\eta_1)$, which is the minimum of $|A|$ linear functions of $\eta_1$. The optimal advertising rule must have $\eta^s_1$ being one of the vertices on one side of $\eta^*_1$, plotted as black dots in the pictures. And $\eta^t_1$ must lie on the other side of $\eta^*_1$, plotted as black segments in the pictures.}
\label{fig:single_binary_app}
\end{figure}

First by Lemma~\ref{lem:opt_size}, there exists an optimal advertising rule that has $|S| = 2$. Let $S = \{s, t\}$.
Without loss of generality assume $R((1,0) > R((0,1))$. Then $R(\eta_1)$ is an increasing function of $\eta_1$.  By Lemma~\ref{lem:opt_decompose}, the optimal advertising rule should not have $\eta^s_1$ or $\eta^t_1$ lying on the left of $\eta_1^*$ and not being a vertex. By Lemma~\ref{lem:opt_never}, the optimal advertising rule should not have both $\eta^s_1$ and $\eta^t_1$ on the right of $\eta_1^*$.  Therefore we must have one of $\eta^s_1$ and $\eta^t_1$ being a vertex on the left of $\eta^*_1$, and the other one on the right of $\eta^*_1$. 

We then show how to compute the the optimal mechanism in $O(|A|^2)$ time. Assume that $\eta^s_1$ is on the left of $\eta^*_1$ and $\eta^t_1$ is on the right of $\eta^*_1$. We enumerate all possibilities of $\eta^s_1$, i.e., the vertices on the left of $\eta^*_1$, and all the linear segments on the right of $\eta^*_1$ that $\eta^t_1$ lies in.  
Denote the linear segment by $[l,r]$. Then  $f(\eta_1) = R(\eta_1) C(\eta_1)$ should be a concave quadratic function of $\eta_1$ in interval $[l,r]$.  Then our problem is equivalent to finding the concave closure of point $(\eta^s_1, f(\eta^s_1))$ and function $f(\eta_1)$ on segment $[l,r]$. This is equivalent to finding the tangent line of $f(\eta_1)$ on segment $[l,r]$ that goes through point  $(\eta^s_1, f(\eta^s_1))$. The tangent point $(\eta^t_1, f(\eta^t_1))$ is either at the endpoints $l,r$, or has 
$$
(f(\eta^t_1) - f(\eta^s_1), \eta^t_1 - \eta^s_1) \propto (f'(\eta^t_1), 1)
$$
which is easy to solve and verify. 

\section{General Problem}

\subsection{Proof of Theorem~\ref{thm:general_hardness}}

We prove Theorem~\ref{thm:general_hardness} by reduction from the \textsc{AlmostColoring} problem from~\cite{khot2012hardness}:
\begin{itemize}
\item[] For any constant $\varepsilon \in (0, \frac{1}{2}]$, and positive integers $k$ and $q$ such that $q \ge 2^k + 1$, given a graph $G(V,E)$, it is NP-hard to distinguish between the following two cases:
\begin{description}
    \item[YES Case:] There are $q$ disjoint independent sets $V_1, \dots, V_q \subseteq V$, such that $|V_i| = (1-\varepsilon) \frac{|V|}{q}$ for $i = 1, \dots, q$.
    \item[NO Case:] There is no independent set in $G$ of size $\frac{|V|}{q^{k+1}}$.
\end{description}
\end{itemize}
For any instance of \textsc{AlmostColoring}, we construct an optimal information advertising  problem whose solution can be used to distinguish the YES Case and the NO Case. Let the state of the world be one of the vertices, i.e., $\Omega = V$ and $n = |V|$. The buyer can possibly have $|V|$ different prior beliefs $\Theta = \{\theta^{(v)}: v\in V\}$. The buyer with prior belief $\theta^{(v)}$ initially thinks that $\omega$ is highly likely to be one of the neighbors of $v$: for some constant $C$
\begin{align*}
   & \theta^{(v)}_j = C, \quad \text{for } j\in V\setminus \mathcal{N}(v),\\
   & \theta^{(v)}_i = qn\cdot C, \quad \text{for } i\in \mathcal{N}(v),
\end{align*}
where $\mathcal{N}(v)$ is the set of the neighboring vertices of $v$ (not including $v$). For simplicity, we call a buyer with prior belief $\theta^{(v)}$ a \emph{type-$v$ buyer}. Assume both the state of the world $\omega$ and the buyer's prior $\theta$ are uniformly distributed, and $\omega$ and $\theta$ are independent, i.e., $\mu(\omega, \theta) = \mu(\omega) \mu(\theta) = \frac{1}{n} \cdot \frac{1}{n}$ for $\omega \in \Omega = V$ and $\theta \in \Theta = \{\theta^{(v)}: v\in V\}$. The buyer can take $2|V|$ actions $A = \{H^{(v)}: v \in V\} \cup \{L^{(v)}: v \in V\}$. Let $d = (1-\varepsilon) \frac{|V|}{q}$ be the size of the subsets in the YES Case. When the buyer takes $H^{(v)}$, he will have utility $M$ when the state of the world is realized to $v$, and have slightly lower utility $M - \frac{M}{d-1}$ when the state of the world is not $v$, i.e.,
\begin{align*}
    u(v, H^{(v)}) = M, \quad u(j, H^{(v)}) = M - \frac{M}{d-1} \text{ for } j\neq v.
\end{align*}
When the buyer takes $L^{(v)}$, he will have zero utility when the state of the world is realized to $v$ and otherwise have utility $M$, i.e.,
\begin{align*}
u(v, L^{(v)}) = 0, \quad u(j, L^{(v)}) = M \text{ for } j \neq v.    
\end{align*}
Then we claim the follows.
\begin{claim} \label{clm:distinguish}
In the YES Case, there exists a mechanism that achieves expected revenue $\ge \frac{M}{n}$. In the NO Case, there exists no mechanism that has expected revenue $> \frac{2M}{qn}$. 
\end{claim}
We first show that in the YES Case, there exists a mechanism that has expected revenue $\frac{M}{n}$. Consider the signaling scheme that reveals which independent set $\omega$ belongs to, i.e., there are $q$ possible signals $S = \{s_1, \dots, s_q\}$ and
\begin{align*}
    \pi(\omega, s_i) = \left\{ \begin{array}{ll}
        1, & \text{ if } \omega\in V_i  \\
        0, & \text{ otherwise}
    \end{array}
    \right.
\end{align*}
Because $V_i$ is an independent set, for any vertex $v \in V_i$, the posterior of a type-$v$ buyer after receiving $s_i$ is the uniform distribution over $V_i$. It is then easy to verify that for any type-$v$ buyer with $v \in V_i$, the cost of uncertainty equals $\frac{M}{d}$ after receiving $s_i$. So if the seller sets a price $\frac{M}{d}$, at least $\frac{|V_i|}{|V|} = \frac{d}{n}$ portion of the buyers will pay for the full revelation of $\omega$ after receiving $s_i$. This holds for all $s_i \in S$. Therefore the seller can have at least $\frac{M}{d} \cdot \frac{d}{n} = \frac{M}{n}$ expected revenue.

We then show that in the NO Case, there exists no mechanism that has expected revenue $> \frac{M}{qn}$. Let's consider an arbitrary mechanism that first sends a signal $s\in S$ using signaling scheme $\pi: \Omega \to \Delta S$, and then sets price $p_s$ if the signal $s$ is realized. Let $V^{(s)}$ be the set of buyer types who will pay for the full revelation of $\omega$ after observing $s$, 
$$
V^{(s)} = \{ v: \text{ the buyer with prior belief } \theta^{(v)} \text{ will pay } p_s \text{ after observing } s\}.
$$
Let's consider two possibilities: (1) $V^{(s)}$ is an independent set; (2) $V^{(s)}$ contains an edge $(i,j)$. 
\paragraph{(1)} If $V^{(s)}$ is an independent set, then it holds that
\begin{itemize}
    \item $|V^{(s)}| \le \frac{|V|}{q^{k+1}}$ because of the NO Case condition.
    \item $p_s \le \frac{M}{d}$. This is because for any probability distribution of the state of the world, $\eta \in \Delta \Omega$, the cost of uncertainty is bounded by $M/d$. More specifically,
\begin{align*}
C(\eta) = & \sum_{\omega=1}^n \eta_\omega \max_{a\in A} u(\omega, a) - \max_{a\in A} \sum_{\omega = 1}^n \eta_\omega u(\omega,a) \\
= &\  M - \max_i \ \max \left\{  \sum_{\omega = 1}^n \eta_\omega u(\omega, H^{(i)}), \sum_{\omega = 1}^n \eta_\omega u(\omega, L^{(i)})\right\}\\
= & \ M - \max_i \ \max \left\{  \eta_i M + (1-\eta_i)(M- \frac{M}{d-1}), \ (1-\eta_i)M\right\}\\
= & \ \min_i\ \min \left\{ (1-\eta_i) \frac{M}{d-1},\ \eta_i M\right\}\\
\le & \ \min \left\{ \left(1-\frac{1}{d}\right) \frac{M}{d-1},\ \frac{1}{d}\cdot M\right\}\\
= & \ \frac{M}{d}.
\end{align*}
\end{itemize}
Therefore the seller's expected revenue conditioning on sending $s$ is no more than $\frac{1}{q^{k+1}}\cdot \frac{M}{d} = \frac{M}{n q^{k}(1-\varepsilon)} \le \frac{2M}{n q^k}$.
\paragraph{(2)} If $V^{(s)}$ contains an edge $(i,j)$, then we claim that $p_s$ cannot exceed $\frac{M}{qn}$. By the definition of $\theta^{(i)}$
and $\theta^{(j)}$, we have 
$$
\theta^{(i)}_j = qn \cdot \theta^{(i)}_i, \quad \theta^{(j)}_i = qn \cdot \theta^{(j)}_j.
$$
WLOG assume $\pi(i, s) \le \pi(j,s)$. Denote by $\eta^{(i)}$ the posterior of a type-$i$ buyer after receiving $s$, then we should have 
$$
\frac{\eta^{(i)}_i}{\eta^{(i)}_j} = \frac{\theta^{(i)}_i}{\theta^{(i)}_j}\cdot \frac{\pi(i,s)}{\pi(j,s)} \le \frac{1}{qn}.
$$
And since $\eta^{(i)}_j \le 1$, it holds that
$$
\eta^{(i)}_i \le \frac{1}{qn},
$$
which means the type-$i$ buyer will believe that the probability of the state of the world $\omega$ being $i$ is no more than $\frac{1}{qn}$ after observing $s$. Then the type-$i$ buyer will not pay more than
\begin{align*}
C(\eta^{(i)}) = &\sum_{\omega=1}^n \eta^{(i)}_\omega \max_{a\in A} u(\omega, a) - \max_{a\in A} \sum_{\omega = 1}^n \eta^{(i)}_\omega u(\omega,a)\\
= &\  M - \max_{a\in A} \sum_{\omega = 1}^n \eta^{(i)}_\omega u(\omega,a)\\
\le & \ M -  \sum_{\omega = 1}^n \eta^{(i)}_\omega u(\omega, L^{(i)}) \\
\le & \ \frac{M}{qn}.
\end{align*}
for the full revelation of $\omega$, which means $p_s \le \frac{M}{qn}$. Therefore the seller's expected revenue conditioning on sending $s$ is no more than $\frac{M}{qn}$.

In both of the cases, the seller's expected revenue conditioning on sending $s$ is no more than $\frac{2M}{qn}$. This holds for all $s$. Therefore in the NO Case, the expected revenue of any mechanism $\le \frac{2M}{qn}$. For any constant $c$, by setting $q = 2c^2$, Claim~\ref{clm:distinguish} implies that it is NP-hard to find $c$-approximation of the optimal mechanism.

\subsection{Efficient Approximation for Binary State} \label{app:eff_bi}

Recall that  $\Lambda(s, p_s) \subseteq \Theta$ is the set of buyer types that would pay for the full revelation of $\omega$, i.e.,
\begin{align} \label{def:lambda_s}
\Lambda(s, p_s) = &\{\theta: C(\eta^s(\theta)) \ge p_s\} 
\end{align}
Let $\mathbf{\Lambda}$ be the set of all possible $\Lambda(s, p_s)$,
\begin{align*}
\mathbf{\Lambda} = \left\{\Lambda(s, p_s): s \text{ is a signal sent by a signaling scheme } \pi, p_s \in \mathbb{R}\right\} 
\end{align*}
We show that when there is a binary state of the world, i.e. $|\Omega| = 2$,  $|\mathbf{\Lambda}|$ is no more than $|\Theta|^2$. This is because 
\begin{align*} 
\Lambda(s, p_s) = &\{\theta: C(\eta^s(\theta)) \ge p_s\} 
\end{align*}
 must be  a convex set. By definition, the posterior belief of a type-$\theta$ buyer after receiving a signal $s$ is equal to 
$$\eta^s (\theta) = \frac{\big( \theta_1 \pi(1, s), \dots, \theta_n \pi(n, s)\big)}{\sum_{\omega=1}^n \theta_\omega \pi(\omega,s)}.$$
And the cost of uncertainty function 
\begin{align*}
C(\eta) & =   \E_{\omega \sim \eta} \big[\max_{a\in A} u(\omega, a) \big] - \max_{a\in A} \E_{\omega \sim \eta} [ u(\omega, a) ] \\
& =   \sum_{\omega=1}^n \eta_\omega \max_{a\in A} u(\omega, a) - \max_{a\in A} \sum_{\omega = 1}^n \eta_\omega u(\omega,a)\\
& = \min_a\, C_a(\eta)
\end{align*}
is the minimum of $|A|$ linear functions. Although the cost of uncertainty function is defined on $\Delta \Omega$, we can naturally extend the domain to $[0,1]^n$ so that 
$$
C(k\eta) = k C(\eta).
$$
Then we have
\begin{align*}
 & \ C(\eta^s(\theta)) \ge p_s \\
 \Longleftrightarrow & \ C\big( \theta_1 \pi(1, s), \dots, \theta_n \pi(n, s)\big) \ge p_s \sum_{\omega=1}^n \theta_\omega \pi(\omega,s).
\end{align*}
The left hand side $C( \theta_1 \pi(1, s), \dots, \theta_n \pi(n, s))$ is a concave function of $\theta$, and the right hand side $p_s \sum_{\omega=1}^n \theta_\omega \pi(\omega,s)$ is a linear function of $\theta$.  Therefore, $\Lambda(s, p_s) = \{\theta: C(\eta^s(\theta)) \ge p_s\}$ must be a convex set. Denote $\theta = (\theta_1, \theta_2)$.  Since $\Lambda(s, p_s) $ is a convex set, the types $\theta$ in $\Lambda(s, p_s)$ must have $\theta_1$ lying in an interval $[L,R]$.  Then as the type space is discrete, we only need to consider the intervals with endpoints $\{\theta_1: \theta \in \Theta\}$ to include all $\Lambda(s, p_s)$. The number of such intervals is no more than $|\Theta|^2$. 

\subsection{Proof of Lemma~\ref{lem:general_eff}} \label{app:general_eff}
Consider an arbitrary optimal advertising rule $\langle S, \pi, \{p_s : s \in S\} \rangle$. Suppose there exist two signals $s,t \in S$ with $(p_s, \Lambda(s, p_s)) = (p_t, \Lambda(t, p_t))$. Then we can merge $s, t$ into one signal $s'$ as follows 
\begin{align*}
&\pi(\omega, s') = \pi(\omega, s) + \pi(\omega, t) \text{ for all } \omega\\
&p_{s'} = p_s = p_t.
\end{align*}
Then according to \eqref{eqn:posterior_eta}, for any buyer of type $\theta \in \Lambda(s, p_s) = \Lambda(t, p_t)$, his posterior after seeing $s'$ is
\begin{align*}
\eta^{s'} (\theta) & = \frac{1}{\sum_{\omega=1}^n \theta_\omega (\pi(\omega,s) + \pi(\omega,t))} \left( \eta^{s} (\theta)\sum_{\omega=1}^n \theta_\omega \pi(\omega,s) +  \eta^{t} (\theta)\sum_{\omega=1}^n \theta_\omega \pi(\omega,t) \right)\\
& = k \eta^{s} (\theta) + (1-k) \eta^{t} (\theta).
\end{align*}
Since the cost of uncertainty function is concave, 
$$
C\big(\eta^{s'} (\theta)\big) \ge k \cdot C\big(\eta^{s} (\theta)\big) + (1-k) \cdot C\big(\eta^{t} (\theta)\big) \ge k p_s + (1-k) p_t = p_{s'}.
$$
So the buyer will still be willing to pay $p_{s'}$. The expected revenue will not decrease.

\section{Optimal Information Disclosure} \label{app:optimal_inf_dis}
\cite{rayo2010optimal} studies the following problem. 
There is a sender endowed with a prospect, which is randomly drawn from a finite set $P=\{1, \dots, N\}$. The probability of $i$ being realized is $p_i > 0$ and $\sum_{i=1}^n p_i = 1$. Each $i\in P$ is characterized by its payoffs $(\pi_i, v_i) \in \mathbb{R}^2$, where $\pi_i$ is the prospect's profitability for the sender, and $v_i$ is  its value to the receiver.

The sender chooses a disclosure rule $\langle \sigma, S \rangle$ to send a signal $s\in S$ drawn from $\sigma(i)$ to the receiver. The receiver observes the signal $s$, and decides whether to ``accept'' ($a=1$) or `` not accept'' ($a=0$). The receiver forgoes an outside option worth $r\in \mathbb{R}$, which is a random variable independent of $i$. So the sender's payoff is $a \cdot \pi$ and the receiver's payoff is $a(v-r)$.

Assume $v\in[0,1]$ and $r \sim \mathcal{U}[0,1]$. Then the sender's expected payoff is 
\begin{align} \label{optimal_info_dis}
\E_{s} \big( \E[\pi|s] \cdot \E[v|s]\big) = \sum_s \left(\sum_i p_i \sigma(i, s) \right) \frac{\sum_i p_i \sigma(i, s) \pi_i}{\sum_i p_i \sigma(i, s)} \cdot \frac{\sum_i p_i \sigma(i, s)v_i}{\sum_i p_i \sigma(i, s)}.
\end{align}

\paragraph{Optimal Information Disclosure as Optimal Advertising.} Consider a Optimal Advertising problem with common buyer prior. Let the state of the world be the realization of the prospect, $\Omega = P$. Let the buyer's common prior equal to the probability distribution of the prospect, i.e., $\theta_i = p_i$. And let the true underlying distribution $\mu$ satisfy $\frac{\mu_i}{\theta_i} \propto \pi_i$, i.e., $\frac{\mu_i}{\theta_i} = \pi_i \cdot M$ where $M$ is a constant so that $\sum_i \theta_i \pi_i \cdot M = 1$. Finally let the cost of uncertainty be the expected value of the prospect, i.e., $C(\theta) = \E_\theta[v] = \sum_{i=1}^n \theta_i v_i$.\footnote{This is not really a valid cost of uncertainty function, as $C(e_i) \neq 0$ for $e_i = (0,\dots,0, 1, 0, \dots, 0)$. But our algorithm still works when the cost of uncertainty is a linear function. } Then it is easy to verify that~\eqref{optimal_info_dis} is equivalent to the optimal advertising problem~\eqref{prog:opt1} with a constant factor difference in the objective function,
\begin{align*}
&\sum_s \left(\sum_i p_i \sigma(i, s) \cdot \frac{\sum_i p_i \sigma(i, s) \pi_i}{\sum_i p_i \sigma(i, s)}\right)  \frac{\sum_i p_i \sigma(i, s)v_i}{\sum_i p_i \sigma(i, s)} \\
= & \sum_s  \left(\sum_i \theta_i \sigma(i, s) \pi_i\right) \frac{\sum_i \theta_i \sigma(i, s)v_i}{\sum_i \theta_i \sigma(i, s)} \\
= & \sum_s  \left(\sum_i \theta_i \sigma(i, s) \cdot \frac{\mu_i}{\theta_i} /M\right)  C(\theta^s) \\
= & \frac{1}{M} \sum_s  \left(\sum_i \mu_i \sigma(i, s) \right)  C(\theta^s).
\end{align*}

\end{document}